\documentclass[a4paper]{article}

\usepackage{hyperref}

\usepackage{geometry}

\usepackage{graphicx}
\usepackage{algorithm}
\usepackage[noend]{algpseudocode}
\usepackage{caption}
\usepackage[protrusion=true, expansion, kerning, spacing, tracking]{microtype}
\usepackage{amsmath}
\usepackage{amssymb}
\usepackage{amsthm}
\usepackage[subnum]{cases}

\newtheorem{theorem}{Theorem}
\newtheorem{lemma}[theorem]{Lemma}
\newtheorem{proposition}[theorem]{Proposition}

\newtheorem{definition}{Definition}

\newcommand{\fig}[3]{%
	\begin{figure}
		\centering
		\includegraphics{#3}
		\caption{#2}
		\label{#1}
	\end{figure}
}

\newcommand{\figX}[3]{%
	\begin{figure}
		\centering
		#3
		\caption{#2}
		\label{#1}
	\end{figure}
}

\newcommand{\algo}[1]{
	\bigskip
	\hrule
	\microtypesetup{tracking=false}
	\begin{algorithmic}
		#1
	\end{algorithmic}
	\microtypesetup{tracking=true}
	\hrule
	\bigskip
}

\newcommand{\p}[2]{\{#1, \ldots, #2\}}
\newcommand{\pij}{\p{i}{j}}
\newcommand{\pji}{\p{j}{i}}
\newcommand{\Matching}{\textsc{Matching}}
\newcommand{\OO}[1]{\mathcal{O}(#1)}
\newcommand{\OA}{\mathcal{A}}
\newcommand{\OB}{\mathcal{B}}
\newcommand{\OC}{\mathcal{C}}
\newcommand{\OD}{\mathcal{D}}
\newcommand{\OL}{\mathcal{L}}
\newcommand{\OH}{\mathcal{H}}
\newcommand{\GR}{\mathcal{G}}
\newcommand{\bn}{\textrm{bn}}

\begin{document}

	\title{Structural Properties of Bichromatic Non-crossing Matchings}

	\author{
		Marko Savi\'{c}\footnote{University of Novi Sad, Faculty of Sciences, Department of Mathematics and Informatics.
				Partly supported by Ministry of Education, Science and Technological Development of the Republic of Serbia (Grant No.~451-03-9/2021-14/200125).
				Partly supported by Provincial Secretariat for Higher Education and Scientific Research, Province of Vojvodina (Grant No.~142-451-3227/2020-01).
			{\tt \{marko.savic, milos.stojakovic\}@dmi.uns.ac.rs}}
		\and
		Milo\v{s} Stojakovi\'{c}\footnotemark[1]
	}

	\maketitle

	\begin{abstract}
		Given a set of $n$ red and $n$ blue points in the plane, we are interested in matching red points with blue points by straight line segments so that the segments do not cross. We develop a range of tools for dealing with the non-crossing matchings of points in convex position. It turns out that the points naturally partition into groups that we refer to as orbits, with a number of properties that prove useful for studying and efficiently processing the non-crossing matchings.

		Bottleneck matching is a matching that minimizes the length of the longest segment. Illustrating the use of the developed tools, we solve the problem of finding bottleneck matchings of points in convex position in $O(n^2)$ time. Subsequently, combining our tools with a geometric analysis we design an $O(n)$-time algorithm for the case where the given points lie on a circle. The best previously known running times were $O(n^3)$ for points in convex position, and $O(n \log n$) for points on a circle.
	\end{abstract}

	\section{Introduction}

	\subsection{Bichromatic non-crossing matchings}
	
	Let $R$ and $B$ be sets of $n$ red and $n$ blue points in the plane, respectively, with $P = R \cup B$ and $R \cap B = \emptyset$. Let $M$ be a perfect matching of points in $R$ to points in $B$ using $n$ straight line segments, that is, each point is an endpoint of exactly one line segment, and each line segment has one red and one blue endpoint. If the line segments do not cross, we refer to such a matching as a bichromatic non-crossing matching.
	
	\subparagraph{Related work}
	
	Geometric non-crossing matchings by straight line segments are widely researched. In case there are two groups of objects and the members of one group are to be matched with the members of the other, we naturally arrive to the bichromatic matchings, often also referred to as the \emph{red-blue matchings}. The examples of real-life problems that fall into this category are numerous, with a whole range of the so-called problems of \emph{supply and demand}, e.g.~matching shoppers and shops, antennas and receivers, etc. A survey by Kaneko and Kano~\cite{kaneko2003discrete} gives an overview of various problems on red and blue points in the plane, including the matching problems. Several papers \cite{aloupis2013non,kratochvil2013non} take a closer look at the algorithms for finding non-crossing planar straight line matchings between red points on one side and various blue objects (in more generality) on the other, devoting particular attention to the special case of blue objects also being points. Note that some geometric versions of the Monge-Kantorovich transportation problem, an optimization problem for matching mines with factories to minimize cost, see~\cite{bogachev2012monge} for a survey, result in crossing-free straight line matchings of mines and factories, with possible multiplicities depending on the weight distribution.
	
	Several papers~\cite{aichholzer2012compatible,aloupis2015bichromatic, aichholzer2018linear} study the collection of \emph{all} possible bichromatic non-crossing straight line red-blue matchings of two given equal-sized planar sets of red and blue points. A number of interesting structural properties of this collection are established, looking into the possibilities to gradually change one matching into the other through a sequence of matchings, such that every pair of consecutive matchings in the sequence has a non-crossing union. A similar problem on monochromatic point sets, where every pair of points is allowed to be matched, has also been looked at, see~e.g.~\cite{aichholzer2009compatible}.
	
	The problem of finding a geometric non-crossing Hamilton path that alternately visits the elements of a given red point set and equally sized blue point set, with certain additional requirements, was studied in~\cite{kaneko2004path}, with an obvious connection to the red-blue matchings (as we can find one in each such Hamilton path). In~\cite{hurtado2008encompassing}, the connections of crossing-free red-blue planar matchings and crossing-free red-blue spanning trees are explored. Bounds on the total number of crossing-free red-blue perfect matchings are given in~\cite{sharir2006number}.

	\subparagraph{Our results}
	
	We take a closer look at the bichromatic non-crossing matchings of points in convex position. In this case, it is straightforward to see that two line segments of a matching cross if and only if their two pairs of endpoints are interleaved in the cyclic order around the convex hull of the given point set. Therefore, the collection of all valid matchings is fully determined by the sequence of red and blue points around the convex hull.
	
	In order to efficiently deal with bichromatic non-crossing matchings on points in convex position we introduce a structure that we refer to as \emph{orbits}, which turn out to capture well the properties of such matchings. As we will show, the points naturally partition into sets, i.e.,~orbits, in such a way that two points of different colors can be connected by a segment in a non-crossing perfect matching if and only if they belong to the same orbit. We go on to study the structure of individual orbits, their properties, as well as the relationship of different orbits of the same point set.
	
	This apparatus enables us to get a grip on the bichromatic non-crossing matchings of points in convex position and work with them in a more efficient manner, with a potential to apply our machinery on the whole range of problems dealing with these matchings. We will present one such application in this article. It is worth noting that another application to several matching optimisation problems recently appeared in~\cite{mantas2021new}.

	\subsection{Bottleneck matchings}
	
	We will illustrate the applicability of our theory of orbits on the problem of efficiently finding the so-called \emph{bottleneck} bichromatic non-crossing matching of points in convex position.
	
	Denote the length of a longest line segment in a straight segment geometric matching $M$ with $\bn(M)$, which we also call the \emph{value} of $M$. We aim to find a perfect 	matching under given constraints that minimizes $\bn(M)$. Any such matching is called a \emph{bottleneck matching} of $P$.

	\subparagraph{Bottleneck matchings -- monochromatic case.}
	
	The monochromatic variant of the problem is the case where points are not assigned colors, and any two points are allowed to be matched.
	
	In \cite{chang1992solving}, Chang, Tang and Lee gave an $O(n^2)$-time algorithm for computing a bottleneck matching of a point set, but allowing crossings. This result was extended by Efrat and Katz in \cite{efrat2000computing} to higher-dimensional Euclidean spaces.
	
	The problem of computing bottleneck monochromatic non-crossing matching of a point set is shown to be NP-complete by Abu-Affash, Carmi, Katz and Trablesi in \cite{abu2014bottleneck}. They also proved that it does not allow a PTAS, gave a $2\sqrt{10}$ factor approximation algorithm, and showed that the case where all points are in convex position can be solved exactly in $O(n^3)$ time. We improved this result in \cite{savic2017faster} by constructing an $O(n^2)$-time algorithm.
	
	
	\subparagraph{Bottleneck matchings -- bichromatic case.}
	
	The problem of finding a bottleneck bichromatic non-crossing matching was proved to be NP-complete by Carlson, Armbruster, Bellam and Saladi in \cite{carlsson2015bottleneck}. But for the version where crossings are allowed, Efrat, Itai and Katz showed in \cite{efrat01geometryhelps} that a bottleneck matching between two point sets can be found in $O(n^{3/2}\log n)$ time.
	
	Biniaz, Maheshwari and Smid in \cite{biniaz2014bottleneck} studied special cases of bottleneck bichromatic non-crossing matchings. They showed that the case where all points are in convex position can be solved in $O(n^3)$ time, utilizing an algorithm similar to the one for monochromatic case presented in \cite{abu2014bottleneck}. They also considered the case where the points of one color lie on a line and all points of the other color are on the same side of that line, providing an $O(n^4)$ algorithm to solve it. The same results for these special cases are independently obtained in \cite{carlsson2015bottleneck}. An even more restricted problem is studied in \cite{biniaz2014bottleneck}, a case where all points lie on a circle, for which an $O(n \log n)$-time algorithm is given.
	
	A variant of the bichromatic case is the so-called bicolored (or multicolored, when there are arbitrary many colors) case, where only the points of the \emph{same} color are allowed to be matched. Abu-Affash, Bhore and Carmi in \cite{abu2017monochromatic} examined bicolored matchings that minimize the number of crossings between edges matching different color sets. They presented an algorithm to compute a bottleneck matching of points in convex position among all matchings that have no crossings of this kind.
	
	\subparagraph{Our results}
	
	Using the orbit theory we solve the problem of finding a bottleneck bichromatic non-crossing matching of points in convex position in $O(n^2)$ time, improving upon the best previously known algorithm of $O(n^3)$-time complexity. Also, combining the same tool set with a geometric analysis we design an optimal $O(n)$ algorithm for the same problem when the points lie on a circle, where the best previously known algorithm has $O(n \log n)$-time complexity.
	
	\subsection{Preliminaries and organization}
	As we deal with bichromatic perfect matchings without crossings, from now on, when we talk about matchings, it is understood that we refer to bichromatic matchings that are both perfect and crossing-free.
	
	Also, we assume that the given points in $P$ are in convex position, i.e.,~they are the vertices of a convex polygon $\mathcal{P}$. Let us label the points of $P$ by $p_0, p_1, \ldots, p_{2n-1}$ in the positive (counterclockwise) direction. To simplify the notation, we will often use only indices when referring to points. We write $\pij$ to represent the set $\{i, i+1, i+2, \ldots, j-1, j\}$. Arithmetic operations on indices are done modulo $2n$. Note that $i$ is not necessarily less than $j$, and that $\pij$ is not the same as $\pji$.
	
	\begin{definition}[Balanced, Blue-heavy, Red-heavy]
		A bichromatic set of points is \emph{balanced} if it contains the same number of red and blue points. If the set has more red points than blue, we say that it is \emph{red-heavy}, and if there are more blue points than red, we call it \emph{blue-heavy}.
	\end{definition}
	
	As we already mentioned, we assume that $P$ consists of $n$ red and $n$ blue points, i.e.,~it is balanced.
	
	The following lemma is a well-known result that ensures the existence of a balanced matching on a point set. A couple of proofs, along with an algorithm that computes one such matching in $O(n \log^2 n)$ time, can be found in~\cite{atallah1985matching}.
	
	\begin{lemma}
		\label{lem:MatchingAlwaysPossible}
		Every balanced set of points admits a matching.
	\end{lemma}
	
	\begin{definition}[Feasible pair]
		We say that $(i,j)$ is a \emph{feasible pair} if there exists a matching containing $(i,j)$.
	\end{definition}
	
	We will make good use of the following characterization of feasible pairs.
	
	\begin{lemma}
		\label{lem:PairFeasible}
		A pair $(i,j)$ is feasible if and only if $i$ and $j$ have different colors and $\p{i}{j}$ is balanced.
	\end{lemma}
	\begin{proof}
		If $(i,j)$ is feasible, then $i$ and $j$ have different colors. Also, there is a matching that contains the pair $(i,j)$, and at the same time the set $\p{i+1}{j-1}$, containing all points on one side of the line $ij$, is matched. Then $\p{i+1}{j-1}$ must be balanced, so $\p{i}{j}$ is balanced as well.
		
		On the other hand, if $i$ and $j$ are of different colors and $\p{i}{j}$ is balanced, then both $\p{i+1}{j-1}$ and $\p{j+1}{i-1}$ are also balanced. Thus we can match $i$ with $j$, and Lemma~\ref{lem:MatchingAlwaysPossible} ensures that each of the sets $\p{i+1}{j-1}$ and $\p{j+1}{i-1}$ can be matched. Clearly, the obtained matching remains crossing-free.
	\end{proof}
	
	The statement of Lemma~\ref{lem:PairFeasible} is quite simple, and we will apply it on many occasions. To avoid its numerous mentions that could make some of our proofs unnecessarily cumbersome, from now on we will use it without explicitly stating it.
	
	The rest of the paper is organized as follows. In Section~\ref{sec:Orbits} we formally define orbits and derive numerous properties that hold for them. We note the existence of a structured relationship between orbits. This leads us to the definition of orbit graphs for which we show certain properties. In Section~\ref{sec:BottleneckMatchingsB} we construct an efficient algorithm for finding a bottleneck matching of points in convex position. For this we follow the general idea from \cite{savic2017faster}, but now we use orbits and their properties for the proofs. In Section~\ref{sec:Circle} we again use properties of orbits and orbit graph to solve the problem of finding a bottleneck matching for points on a circle in $O(n)$ time.

	\section{Orbits and their properties}
	\label{sec:Orbits}

	\begin{definition}[Functions $o^+$ and $o^-$]
	By $o^+, o^- : P \rightarrow P$ we denote functions, such that $o^+(i)$ is the first point $j$ starting from $i$ in the positive direction with $(i,j)$ being feasible, and $o^-(i)$ is the first point $j$ starting from $i$ in the negative direction with $(i,j)$ being feasible.
	\end{definition}
	
	As $P$ is balanced, Lemma~\ref{lem:MatchingAlwaysPossible} guarantees that both $o^+$ and $o^-$ are well-defined.
	
	

	\begin{proposition}
		\label{prp:OGoingIn}
		If a set $\pij$ is such that the number of points in $\pij$ of the same color as $i$ is not larger than the number of points of the other color, then $o^+(i) \in \p{i+1}{j}$.
		
		If a set $\pij$ is such that the number of points in $\pij$ of the same color as $j$ is not larger than the number of points of the other color, then $o^-(j) \in \p{i}{j-1}$.
	\end{proposition}
	\begin{proof}
		W.l.o.g.~assume that $i$ is red. We observe the difference between the number of red points and the number blue points in $\p{i}{k}$, as $k$ goes from $i$ to $j$. In the beginning, when $k = i$, this difference is $1$, and at the end, when $k = j$ the difference is at most 0. In each step this difference changes by $1$, so the first time this difference is $0$, the point $k$ must be blue. This is the first time the set $\p{i}{k}$ is balanced, and hence $o^+(i) = k \in \p{i+1}{j}$.
		
		The second part of the proposition is proven analogously.
	\end{proof}
	
	A straightforward consequence of Proposition~\ref{prp:OGoingIn} follows.
	
	\begin{proposition}
		\label{prp:BalancedO}
		If $\pij$ is balanced, then $o^+(i) \in \p{i+1}{j}$ and $o^-(j) \in \p{i}{j-1}$. $\hfill \Box$
	\end{proposition}
	
	The next proposition establishes the connection between $o^+$ and $o^-$.
	
	\begin{proposition} \label{prp:O-1BijectionInverse}
		Functions $o^+$ and $o^-$ are bijective, and they are inverses of each other.
	\end{proposition}
	\begin{proof}
		It is enough to prove that, for all $i \in P$, we have $o^+(o^-(i)) = i$ and $o^-(o^+(i)) = i$.
		
		Let $j = o^+(i)$ and $k = o^-(j)$. Suppose that $i \neq k$. By definition of $o^+$, the set $\pij$ is balanced, so by Proposition~\ref{prp:BalancedO} we have that $k \in \p{i}{j-1}$. On the other hand, by definition of $o^-$, the set $\p{k}{j}$ is also balanced, so $\p{i}{k-1}$ must be balanced as well. But this means, again by Proposition~\ref{prp:BalancedO}, that $o^+(i) \in \p{i+1}{k-1}$, which is a contradiction. Hence, $o^-(o^+(i)) = i$. The claim that $o^+(o^-(i)) = i$ is proven analogously.
	\end{proof}
	
	This proposition allows us that instead of using both $o^+$ and $o^-$, we just use the function $o := o^+$, along with the standard notation $o^k, k \in \mathbb{Z}$, to represent the repeated composition of $o$ with itself, where $o$ appears $k$ times. By definition, $o^0$ is the identity function on $P$, and expressions with a negative integer in the exponent are evaluated using the inverse function, $o^{-k} = (o^-)^k$.

	Now we are ready to define orbits.
	
	\begin{definition}[Orbit]
		The \emph{orbit} of $i$, denoted by $\OO{i}$, is defined by $\OO{i} := \{o^k(i) : k \in \mathbb{Z}\}$. By $\OO{P}$ we denote the set of all orbits of the set $P$ of points in convex position, that is, $\OO{P} := \{\OO{i} : i \in P\}$.
	\end{definition}
	
	\fig{fig:Orbits}{Orbits -- an example.}{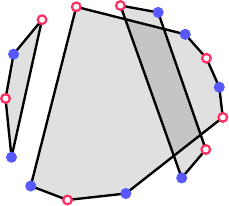}%
	An example of a balanced 2-colored set of points in convex position, along with its set of orbits can be found in Figure~\ref{fig:Orbits}. Note that from the definition of orbits it is clear that for each $j \in \OO{i}$ we have $\OO{j} = \OO{i}$, and thus the set of all orbits is a \emph{partition} of the set of all points.
	
	It is not hard to convince oneself that the number of orbits can be anything from $1$, when colors alternate, as in Figure~\ref{fig:OrbitsNumber}(a), to $n/2$, when points in each color group are consecutive, as in Figure~\ref{fig:OrbitsNumber}(b).
	
	\figX{fig:OrbitsNumber}{(a) One orbit of size $2n$.(b) $n$ orbits of size $2$.}{%
		\parbox{4cm}{\centering\includegraphics{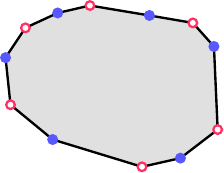}\\(a)}
		\hspace{1cm}%
		\parbox{4cm}{\centering\includegraphics{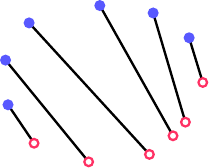}\\(b)}
	}

	Next, we prove a number of properties of orbits.
	
	The first proposition provides a simple characterization of a feasible pair via orbits, which is essential for our further application of orbits.
	
	\begin{proposition}
		\label{prp:OrbitsFeasible}
		Points $i$ and $j$ form a feasible pair if and only if they have different colors and $\OO{i} = \OO{j}$.
	\end{proposition}
	\begin{proof}
		First, suppose that $i$ and $j$ have different colors and belong to the same orbit. Then $j = o^s(i)$, where $s$ is odd (as $i$ and $j$ have different colors). For each $r \in \{0, \ldots, s-1\}$, the pair $(o^r(i), o^{r+1}(i))$ is feasible so $\p{o^r(i)}{o^{r+1}(i)}$ is balanced. This, together with the fact that the sequence $o^0(i), o^1(i), ..., o^s(i)$ alternates between red and blue points, implies that $\pij$ is balanced as well, that is, the pair $(i,j)$ is feasible.
		
		\fig{fig:OrbitsFeasibleProof}{Illustrating the proof of Proposition~\ref{prp:OrbitsFeasible}.}{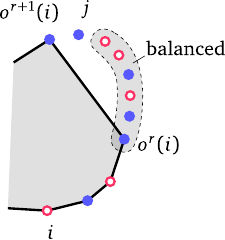}%
		Next, let $(i,j)$ be a feasible pair, where, say, $i$ is red and $j$ is blue. Suppose for a contradiction that $i$ and $j$ belong to different orbits. Let $r$ be such that $j \in \p{o^r(i)+1}{o^{r+1}(i)-1}$, see Figure~\ref{fig:OrbitsFeasibleProof}. W.l.o.g.~suppose that $o^r(i)$ is blue (the other case is symmetrical with respect to the direction around $P$). Since both $(i, o^r(i))$ and $(i, j)$ are feasible pairs, then $\p{o^r(i) + 1}{j}$ is balanced. The points $o^r(i)$ and $j$ are of the same color, so $\p{o^r(i)}{j-1}$ is also balanced. However, Proposition~\ref{prp:BalancedO} implies that $o^{r+1}(i) = o(o^r(i)) \in \p{o^r(i)+1}{j-1}$, which contradicts the choice of $r$.
	\end{proof}
	
	The following proposition discusses the way a feasible pair divides an orbit, whether it belongs to it or not.
	
	\begin{proposition}
		\label{prp:FeasibleSplitsBalanced}
		A feasible pair divides points of \emph{any} orbit into two balanced parts.
	\end{proposition}
	\begin{proof}
		Let $(i, j)$ be a feasible pair and let $\OA$ be an orbit. By Proposition~\ref{prp:OrbitsFeasible} points can be matched only within their orbit, so if $\pij \cap \OA$ is not balanced, then it is not possible to complete a matching containing $(i, j)$ which contradicts $(i,j)$ being feasible.
	\end{proof}

	Informally speaking, the following proposition ensures that by repeatedly applying function $o$, we follow the points of an orbit as they appear on $\mathcal{P}$, thus visiting \emph{all} the points of the orbit in a \emph{single} turn around the polygon.
	
	\begin{proposition}
		\label{prp:SingleTurn}
		For every point $i$, no point of $\OO{i}$ lies between $i$ and $o(i)$, that is, $\p{i}{o(i)} \cap \OO{i} = \{i,o(i)\}$.
	\end{proposition}
	\begin{proof}
		Suppose there is a point $j \in \OO{i}$ such that $j \in \p{i}{o(i)} \setminus \{i,o(i)\}$. The colors of $i$ and $o(i)$ are different, so the color of $j$ is either different from $i$ or from $o(i)$.
		
		If $i$ and $j$ have different colors, knowing that they belong to the same orbit, by Proposition~\ref{prp:OrbitsFeasible} the pair $(i,j)$ is feasible, which contradicts $o(i)$ being the first point from $i$ in the positive direction such that $(i,o(i))$ is feasible.
		
		The case when $o(i)$ and $j$ have different colors is treated analogously.
	\end{proof}
	
	The following two propositions are simple consequences of the previous statement.
	
	\begin{proposition}
		\label{prp:NeighboringDifferentColors}
		Any two neighboring points in an orbit have different colors.
	\end{proposition}
	\begin{proof}
		From Proposition~\ref{prp:SingleTurn} we have that if $i$ and $j$ are neighboring points on an orbit, then either $j = o(i)$ or $i = o(j)$. By the definition of the function $o$, this means that $i$ and $j$ have different colors.
	\end{proof}
	
	\begin{proposition}
		\label{prp:OrbitBalanced}
		Every orbit is balanced.
	\end{proposition}
	\begin{proof}
		This follows directly from Proposition~\ref{prp:NeighboringDifferentColors}.
	\end{proof}
	
	Next, we discuss a structural property of two different orbits.
	
	\begin{proposition}
		\label{prp:NeighborsSameColor}
		Let $i$ and $j$ be points from two different orbits such that there are no other points from their orbits between them, that is, $\pij \cap \OO{i} = i$ and $\pij \cap \OO{j} = j$. Then, $i$ and $j$ have the same color.
	\end{proposition}
	\begin{proof}
		Suppose for a contradiction that $i$ and $j$ have different colors, say, $i$ is blue and $j$ is red. Since they are not from the same orbit, by Proposition~\ref{prp:OrbitsFeasible} the pair $(i,j)$ is not feasible. Thus, $\pij$ is not balanced, so it is either red-heavy or blue-heavy.
		
		If it is red-heavy, then by Proposition~\ref{prp:OGoingIn} we have $o(i) \in \p{i+1}{j}$, which contradicts $\pij \cap \OO{i} = i$.
		
		If $\pij$ is blue-heavy, then, again by Proposition~\ref{prp:OGoingIn}, $o^{-1}(j) \in \p{i}{j-1}$, which contradicts $\pij \cap \OO{j} = j$.
	\end{proof}

	Moving on to the algorithmic part of the story, we show that we can efficiently compute all the orbits, or more precisely -- all the values of the function $o$.
	
	\begin{lemma}
		\label{lem:OrbitsComplexity}
		The function $o(i)$, for all $i$, can be computed in $O(n)$ time.
	\end{lemma}
	\begin{proof}
		The goal is to find $o(i)$ for each $i \in P$. We start by showing that there is an $i_0 \in P$ such that for every $j \in P$, we have that $\p{i_0}{j}$ is either balanced or red-heavy, and that we can compute such $i_0$ in $O(n)$ time.
		
		We define $z_i$ to be the number of red points minus the number of blue points in $\p{0}{i-1}$. All these values can be computed in $O(n)$ time, since $z_{i} = z_{i-1} \pm 1$, where we take the plus sign if the point $i-1$ is red, and the minus sign if it is blue. For $i_0$ we take $i$ for which $z_i$ is minimum, breaking ties arbitrarily. It is straightforward to check that the above condition is satisfied: if there were a $j$ such that $\p{i_0}{j}$ is blue-heavy, then $z_j$ would have been less than $z_{i_0}$, which is impossible due to the way we selected $i_0$.
		
		To compute the function $o$ in all the red points, we run the following algorithm.
		
		\algo{
			\State Find $i_0$ as described.
			\State Create new empty stack $\mathcal{S}$.
			\For {$i \in \p{i_0}{i_0-1}$}
			\If {$i \in R$}
			\State $\mathcal{S}.Push(i)$
			\Else
			\State $j \leftarrow \mathcal{S}.Pop()$
			\State $o(j) \leftarrow i$
			\EndIf
			\EndFor
		}
		
		The way $i_0$ is chosen guarantees that for every $j \in P$, the number of blue points in the set $\p{i_0}{j}$ is at most the number of red points in the same set, i.e.,~the set is either balanced or red-heavy. This ensures that the stack will never be empty when Pop operation is called. When $o(j)$ is assigned, the point $j$ is the last on the stack because each red point that came after $j$ is popped when its blue pair is encountered, meaning that $\p{j}{i}$ is balanced. Moreover, this is the first time such a situation happens, so the assignment $o(j) = i$ is correct.
		
		By running this algorithm we computed the function $o$ in all red points. To compute it in blue points as well, we run an analogous algorithm where the color roles are swapped.
		All the parts of this process run in $O(n)$ time, so the function $o$ and, thereby, all orbits, are computed in $O(n)$ time as well.
	\end{proof}

	We define two categories of feasible pairs according to the relative position within their orbit.
	
	\begin{definition}[Edge, Diagonal]
		We call a feasible pair $(i,j)$ an \emph{edge} if and only if $i = o(j)$ or $j = o(i)$; otherwise, it is called a \emph{diagonal}.
	\end{definition}
	
	In other words, pairs consisting of two neighboring vertices of an orbit are edges, and all other feasible pairs are diagonals. Note that edges are not necessarily neighboring vertices in $P$.
	

	\begin{proposition}
		\label{prp:MatchingWithEdgesBalanced}
		If $\pij$ is balanced, then points in $\pij$ can be matched using edges only.
	\end{proposition}
	\begin{proof}
		We prove this by induction on the size of $\pij$. The statement obviously hold for the base case, where $j = i+1$, since $(i,i+1)$ itself must be an edge.
		
		Let us assume that the statement is true for all balanced sets of points of size less than $r$, and let $|\pij| = r$. Proposition~\ref{prp:BalancedO} implies that $o(i) \in \pij$. We construct a matching on $\pij$ by taking the edge $(i,o(i))$, and edge-only matchings on $\p{i+1}{o(i)-1}$ and $\p{o(i)+1}{j}$, which are provided by the induction hypothesis.
	\end{proof}

	When we speak about edges, we consider them as ordered pairs of points, so that the edge $(i, o(i))$ is considered to be directed from $i$ to $o(i)$.  We say that points $\p{i}{o(i)} \setminus \{i,o(i)\}$ lie on the right side of that edge, and points $\p{o(i)}{i} \setminus \{i, o(i)\}$ lie on its left side. Directionality and coloring together imply two possible types of edges, as the following definition states.
	
	\begin{definition}[Red-blue edge, Blue-red edge]
		We say that $(i,o(i))$ is a \emph{red-blue} edge if $i \in R$, and \emph{blue-red} edge if $i \in B$.
	\end{definition}
	
	Note that sometimes an orbit comprises only two points, in case when $o(o(i)) = i$; we think of it as if it has two edges, $(i, o(i))$ and $(o(i),i)$, one being red-blue and the other being blue-red.
	
	\begin{proposition}
		\label{prp:NoEdgesOfTheSameTypeCross}
		Two edges of the same type (both red-blue, or both blue-red) from different orbits do not cross.
	\end{proposition}
	
	
	\begin{proof}
		Let $(i, o(i))$ and $(j,o(j))$ be two edges of the same type, and $\OO{i} \neq \OO{j}$. Suppose, for a contradiction, that these edges cross, then we either have $j \in \p{i}{o(i)}$ or $i \in \p{j}{o(j)}$.
		
		W.l.o.g.~we can assume that $j \in \p{i}{o(i)}$. Then there are no points from $\OO{i} \cup \OO{j}$ in $\p{j}{o(i)}\setminus\{j,o(i)\}$, and Proposition~\ref{prp:NeighborsSameColor} implies that points $o(i)$ and $j$ have the same color, but this contradicts the assumption that $(i,o(i))$ and $(j,o(j))$ are of the same type.
	\end{proof}

	\begin{proposition}
		\label{prp:OrbitInteractionParity}
		For every two orbits $\OA, \OB \in \OO{P}$, $\OA \neq \OB$, either all points of $\OB$ are on the right side of red-blue edges of $\OA$, or all points of $\OB$ are on the right side of blue-red edges of $\OA$.
	\end{proposition}

	\begin{proof}
		\fig{fig:OrbitInteractionParityProof}{Illustrating the proof of Proposition~\ref{prp:OrbitInteractionParity}.}{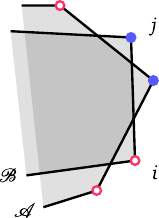}%
		Suppose for a contradiction that there are two points from $\OB$, one on the right of a red-blue edge of $\OA$, and the other on the right of a blue-red edge of $\OA$, see Figure~\ref{fig:OrbitInteractionParityProof}. Let $i$ and $j$ be two such points with no other points from $\OB$ in $\pij$ (we can always find such a pair, since each point of $\OB$ is either behind a red-blue edge, or behind a blue-red edge of $\OA$). Then, $(i,j)$ is an edge of $\OB$ which crosses both a red-blue edge and a blue-red edge of $\OA$, which contradicts Proposition~\ref{prp:NoEdgesOfTheSameTypeCross}.
	\end{proof}
	
	The following proposition tells us about how the orbits are mutually synchronized.
	
	\begin{proposition}
		\label{prp:OrbitSynchronicity}
		Let $\OA, \OB \in \OO{P}$. There are no points of $\OB$ on the right side of red-blue edges of $\OA$ if and only if there are no points of $\OA$ on the right of blue-red edges of $\OB$.
	\end{proposition}
	\begin{proof}
		\fig{fig:OrbitSynchronicityProof}{Illustrating the proof of Proposition~\ref{prp:OrbitSynchronicity}.}{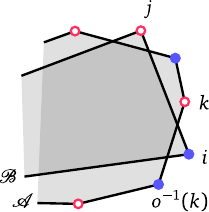}%
		If $\OA = \OB$ this is trivially true.
		
		Assume that there is no point of $\OB$ on the right side of a red-blue edge of $\OA$. Suppose for a contradiction that there is a blue-red edge $(i,j)$ of $\OB$ such that there are points of $\OA$ on its right side, see Figure~\ref{fig:OrbitSynchronicityProof}. Let $k$ be the first point from $\OA$ in $\pij$, observed in the positive direction around $\mathcal{P}$ starting from $i$. It must be red, otherwise point $i$ of $\OB$ would be on the right side of the red-blue edge $(o^{-1}(k),k)$ of $\OA$.
		But now, $i \in \OB$ is blue and $k \in \OA$ is red, and no points of $\OA\cup \OB$ are in $\p{i}{k}$ other than $i$ and $k$, which contradicts Proposition~\ref{prp:NeighborsSameColor}.
		
		The other direction is proven analogously.
	\end{proof}

	\begin{definition}[Relation $\preceq$]
		We define relation $\preceq$ on $\OO{P}$ by setting $\OA \preceq \OB$ if and only if there are no points of $\OB$ on the right sides of red-blue edges of $\OA$ (which, by Proposition~\ref{prp:OrbitSynchronicity}, is equivalent to no points of $\OA$ being on the right sides of blue-red edges of $\OB$).
	\end{definition}

	\begin{proposition}
		\label{prp:PrecInTotalOrder}
		The relation $\preceq$ on $\OO{P}$ is a total order.
	\end{proposition}
	\begin{proof}
		For each $\OA, \OB \in \OO{P}$, the following holds.
		
		\textbf{Totality.}
		$\OA \preceq \OB$ or $\OB \preceq \OA$.
		
		If $\OA = \OB$ this is trivially true. Suppose $\OA \preceq \OB$ does not hold. Because of Proposition~\ref{prp:OrbitInteractionParity}, no points of $\OB$ are on the right side of blue-red edges of $\OA$, so $\OB \preceq \OA$, by the definition of the relation $\preceq$.
		
		\textbf{Antisymmetry.}
		If $\OA \preceq \OB$ and $\OB \preceq \OA$, then $\OA = \OB$.
		
		From $\OA \preceq \OB$ we know that no points of $\OA$ are on the right side of blue-red edges of $\OB$. But, since  $\OB \preceq \OA$, there are no points of $\OA$ on the right side of red-blue edges of $\OB$, either. This is only possible if $\OA = \OB$.
		
		\textbf{Transitivity.}
		If $\OA \preceq \OB$ and $\OB \preceq \OC$, then $\OA \preceq \OC$.
		
		If $\OA \preceq \OB$ then all red-blue edges of $\OA$ must lie on the right side of red-blue edges of $\OB$, because no red-blue edges of $\OA$ can cross a red-blue edge of $\OB$ (Proposition~\ref{prp:NoEdgesOfTheSameTypeCross}) and there are no points of $\OA$ on the right side of blue-red edges of $\OB$. But, since $\OB \preceq \OC$, there are no points of $\OC$ right of red-blue edges of $\OB$, so no point of $\OC$ can be on the right side of some red-blue edge of $\OA$. Hence, $\OA \preceq \OC$.
	\end{proof}

	\begin{proposition}
		\label{prp:Consecutive}
		Let $\OA$ and $\OB$, $\OA \preceq \OB$, be two consecutive orbits in the total order of orbits, that is, there is no $\OL$ different from $\OA$ and $\OB$, such that $\OA \preceq \OL \preceq \OB$. If $i$ and $j$ are two points, one from $\OA$ and the other from $\OB$ such that there are no points from $\OA$ or $\OB$ in $\pij$ other than $i$ and $j$, then $i$ and $j$ are two consecutive points on $\mathcal{P}$.
		
		The converse also holds, for any two consecutive points $i$ and $i+1$ in $P$ which belong to different orbits, orbits $\OO{i}$ and $\OO{i+1}$ are two consecutive orbits in the total order of orbits.
	\end{proposition}

	Note that Proposition~\ref{prp:NeighboringDifferentColors} and Proposition~\ref{prp:NeighborsSameColor} ensure that two consecutive points in $P$ belong to different orbits if and only if they have the same color.
	
	\begin{proof}(of Proposition~\ref{prp:Consecutive})
		\fig{fig:ConsecutiveProof}{Illustrating the proof of Proposition~\ref{prp:Consecutive}.}{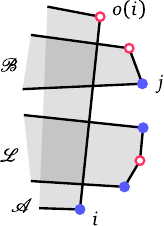}%
		Assume that $i \in \OA$ and $j \in \OB$ are two points such that $\OA \cap \pij = \{i\}$ and $\OB \cap \pij = \{j\}$, see Figure~\ref{fig:ConsecutiveProof}. (The case when $i \in \OB$ and $j \in \OA$ is proven analogously.)
		
		Points $i$ and $j$ must have the same color, by Proposition~\ref{prp:NeighborsSameColor}. Since $j$ is on the right side of the edge $(i,o(i))$ and $\OA \preceq \OB$, that edge must be blue-red, so both $i$ and $j$ are blue.
		
		Suppose that there is an orbit $\OL$ with points in $\p{i+1}{j-1}$. But then, those points are on the right side of the blue-red edge $(i,o(i))$ and on the right side of the red-blue edge $(o^{-1}(j),j)$, that is, $\OA \preceq \OL$ and $\OL \preceq \OB$, a contradiction.
		
		To show the converse statement, assume that points $i$ and $i+1$ belong to different orbits. W.l.o.g., assume $\OO{i} \preceq \OO{i+1}$. If there is an orbit $\OL$ different from both $\OO{i}$ and $\OO{i+1}$, such that $\OO{i} \preceq \OL \preceq \OO{i+1}$, then $i$ would lie on the right side of red-blue edges of $\OL$, and no points of $\OO{i+1}$ would lie on the right side of red-blue edges of $\OL$. But, this is not possible since the position of points $i$ and $i+1$ must be the same relative to any edge containing neither $i$ nor $i+1$.
	\end{proof}

	\subsection{Orbit graphs}
	\label{sec:OrbitGraphs}
	
	\begin{definition}[Orbit graph]
		\emph{Orbit graph} $\GR(P)$ is a directed graph whose vertex set is the set of orbits $\OO{P}$, and there is an arc from an orbit $\OA$ to an orbit $\OB$, that is, $(\OA, \OB) \in E(\GR(P))$, if and only if $\OA$ and $\OB$ cross each other and $\OA \preceq \OB$.
	\end{definition}
	\begin{proposition}
		\label{prp:ForbidenInducedSubgraph}
		Let $\OA, \OB, \OC \in \OO{P}$. If both $(\OA,\OB)$ and $(\OA,\OC)$ are arcs of $\GR(P)$, or both $(\OB,\OA)$ and $(\OC,\OA)$ are arcs of $\GR(P)$, then either $(\OB,\OC)$ is an arc of $\GR(P)$, or $(\OC,\OB)$ is an arc of $\GR(P)$.
	\end{proposition}
	\begin{proof}
		\fig{fig:ForbidenInducedSubgraphProof}{Illustrating the proof of Proposition~\ref{prp:ForbidenInducedSubgraph}.}{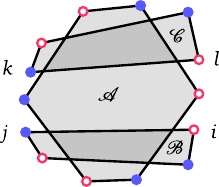}%
		Assume that in $\GR(P)$ there is an arc between $\OA$ and $\OB$, an arc between $\OA$ and $\OC$, but no arc between $\OB$ and $\OC$. By definition, $\OA$ crosses both $\OB$ and $\OC$, and $\OB$ and $\OC$ do not cross, as illustrated in Figure~\ref{fig:ForbidenInducedSubgraphProof}. Then, there is an edge $(i,j)$ of $\OB$ such that the whole $\OC$ lies on its right side, and there is an edge $(k,l)$ of $\OC$ such that the whole $\OB$ lies on its right side.
		
		From Proposition~\ref{prp:NeighborsSameColor} we know that points $i$ and $l$ must be of the same color. Therefore, edges $(i,j)$ and $(k,l)$ are of different types. Orbit $\OA$ crosses both $\OB$ and $\OC$, so it must cross both $(i,j)$ and $(k,l)$. If $\OA \preceq \OB$ then $(i,j)$ must be red-blue, since there are points of $\OA$ on the right side of $(i,j)$, and thus $(k,l)$ is blue-red. But there are also points of $\OA$ on the right side of $(k,l)$, so $\OC \preceq \OA$. Analogously, If $\OB \preceq \OA$, then $\OA \preceq \OC$.
		
		Hence, if both $\OA \preceq \OB$ and $\OA \preceq \OC$ or both $\OB \preceq \OA$ and $\OC \preceq \OA$, then $\OB$ and $\OC$ must cross.
	\end{proof}

	\begin{proposition}	
		\label{prp:connectedInBetween}
		Let $\OA, \OB, \OC \in \OO{P}$, and $\OA \preceq \OB \preceq \OC$. If $(\OA, \OC) \in E(\GR(P))$, then $(\OA, \OB)\in E(\GR(P))$ and $(\OB, \OC) \in E(\GR(P))$.
	\end{proposition}
	
	\begin{proof}
		From Proposition~\ref{prp:ForbidenInducedSubgraph} we know that if either $(\OA,\OB)$ or $(\OB,\OC)$ is an arc of $\GR(P)$, then both of them must be. Therefore, let us assume for the sake of contradiction that neither $(\OA,\OB)$ nor $(\OB,\OC)$ is an arc of $\GR(P)$. That means that $\OA$ and $\OC$ intersect, but $\OA$ and $\OB$ do not intersect and $\OB$ and $\OC$ do not intersect.
		
		Since $\OA \preceq \OB$, and $\OA$ and $\OB$ do not intersect, there is a red-blue edge of $\OB$ such that all the points of $\OA$ are on its right side. Similarly, since $\OB \preceq \OC$, and $\OB$ and $\OC$ do not intersect, there is a blue-red edge of $\OB$ such that all the points of $\OC$ are on its right side. This contradicts the requirement that $\OA$ and $\OC$ intersect.
	\end{proof}
	
	\begin{proposition}
		\label{prp:segmentedStructure}
		Let $\OA, \OB, \OC, \OD \in \OO{P}$ such that $\OA \preceq \OB \preceq \OC \preceq \OD$. If $(\OA, \OD) \in E(\GR(P))$, then $(\OB, \OC) \in E(\GR(P))$.
	\end{proposition}
	
	\begin{proof}
		We apply Proposition~\ref{prp:connectedInBetween} once to $\OA$, $\OC$ and $\OD$ to conclude that $(\OA, \OC) \in E(\GR(P))$, and than again to $\OA$, $\OB$ and $\OC$ to conclude that $(\OB, \OC) \in E(\GR(P))$.
	\end{proof}
	
	The previous two propositions imply that all orbits ``under'' any arc form (an orientation of) a clique.

	\begin{definition}[Segmented graph, Segment]
		We say that a directed graph $G$ is \emph{segmented} if its vertices can be labeled $v_0, v_1, \ldots, v_{m-1}$ so that $(v_i,v_j) \in E(G)$ implies $i \leq j$ and $(v_s,v_t) \in E(G)$ for all $0 \leq i \leq s \leq t \leq j < m$.
		
		For a segmented graph $G$, an arc $(v_s,v_t) \in E(G)$ is called a \emph{segment} if there is no arc $(v_i,v_j) \in E(G)$ other than $(v_s, v_t)$ such that $i \leq s \leq t \leq j$.
	\end{definition}
	
	Notice that any segmented graph is fully defined by the enumeration of its vertices and the set of all of its segments.
	
	Proposition~\ref{prp:segmentedStructure} tells us that any orbit graph is segmented. Interestingly, the converse is also true, as we will show next.
	
	\begin{theorem}[A characterization of orbit graphs]
		\label{prp:segmentedEqOrbitGraph}
		
		A directed graph $G$ is the orbit graph of some bichromatic set of points in convex position if and only if it is segmented.
	\end{theorem}
	
	\begin{proof}
		The ``only if'' part follows from Proposition~\ref{prp:segmentedStructure}.
		
		For the ``if'' part, let $G$ be a segmented graph with the vertex set $\{ v_0, v_1, \ldots, v_{m-1}\}$, and segments $(v_{s_k}, v_{t_k})$, for $k \in \{0, \ldots, Z-1 \}$, where $Z$ is the number of segments, and the sequence $s$ is increasing, that is, $s_e < s_f$ for $0 \leq e < f < Z$. The sequence $t$ must be increasing as well, as by the definition of a segment, no segment can be "nested" inside another. Note also that no two segments can share the starting vertex, nor the ending vertex.
		
		Our goal is to construct a set $P$ of $2n$ bichromatic points in convex position, for a suitable integer $n$, whose orbit graph is isomorphic to $G$. This will be done in several steps, what follows is the outline of the rest of the proof.
		
		We will start by constructing an auxiliary sequence $Q$ of vertices of $G$ of length $2n$, in which each vertex of $G$ appears at least once. Then, using $Q$, we will construct a sequence $C$ of colors (\emph{red} or \emph{blue}) of length $2n$. These will be the colors of $P$, in sequence around the convex hull. Note that this sequence of points will have $2n$ elements, same as $Q$, giving the obvious bijective correspondence between the points in $P$ and the elements of $Q$. To complete the proof, we first need to show that two points of $P$ belong to the same orbit if and only if their corresponding elements of $Q$ are the same. Then, once we prove that two orbits of $P$ cross if and only if their corresponding elements of $Q$ are connected with an arc in $G$, we are done.
		
			\fig{fig:SegmentedExample}{The construction of the sequence $Q$ for a given segmented graph. The example graph is on the left, where only the vertices and the segments are visible (other arcs are not shown). If $Q_i = v_j$, a point is plotted at the position $(i,j)$ with color $C_i$.}{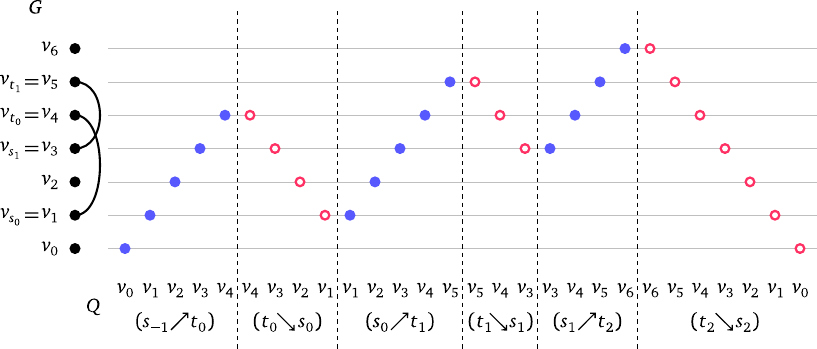}%
		
		For $i \leq j$, we call $(i {\nearrow} j) := ( v_i, v_{i+1}, \ldots, v_{j})$ an \emph{increasing sequence}, and $(j {\searrow} i) := ( v_j, v_{j-1}, \ldots, v_{i})$ a \emph{decreasing sequence}.	
		The sequence $Q$ of vertices of $G$ is constructed concatenating increasing and decreasing sequences as follows
		$$Q = \bigoplus_{e=0}^{Z} \big( (s_{e-1}{\nearrow}{t_e}) \oplus (t_{e}{\searrow}{s_e}) \big),$$
		where the symbol $\oplus$ represents the sequence concatenation operation, $s_{-1} :=0$, $s_Z := 0$ and $t_Z := m-1$. The construction of the sequence $Q$ is visualized in Figure~\ref{fig:SegmentedExample}. Note that the last element of an increasing (resp.~decreasing) part of $Q$ is always the same as the first element of the following decreasing (resp.~increasing) part. Hence, two consecutive elements of $Q$ are the same if and only if they lie on the ``switch'' from increasing to decreasing, or vice versa. It is straightforward to check that the length of $Q$ is even, and the value of $n$ is set to be equal to the half of the length of $Q$.
		
		The color sequence $C$ is defined by $C_0 = \textit{blue}$, and $C_i = C_{i-1}$ if $Q_{i}\neq Q_{i-1}$, and $C_i = \overline{C_{i-1}}$ if $Q_{i} = Q_{i-1}$, for $i \in \{1, \ldots, 2n-1\}$, where $\overline{c}$ denotes the color other than $c$. This means that the increasing parts of $Q$ are all blue, and the decreasing parts of $Q$ are all red. Hence, we have
		$$C = \bigoplus_{e=0}^{Z} \big( \textit{blue}^{t_e-s_{e-1}+1} \oplus \textit{red}^{t_e-s_e+1} \big). $$
		
		Let $P$ be the set of $2n$ points in convex position, where $i$-th point has the color $c_i$. We proceed to show that two points of $P$ belong to the same orbit if and only if their corresponding elements of $Q$ are the same.
		
		
		Let $I(\ell, color)$ be $1$ if $C_\ell = color$ and $0$ otherwise. Further, we define $I(\ell, color_1, color_2) = I(\ell, color_1)I(\ell+1, color_2)$. Let $\text{cnt}(i,j,color)$, for $0 \leq i < j < 2n$, be the number of points in $\p{i}{j}$ in $color$, that is, $\text{cnt}(i,j,color) = \sum_{\ell=i}^{j}I(\ell,color)$.
		

		Now we can write
		$$\text{cnt}(i,j,blue) = \sum_{\ell=i}^{j-1}\Big(I(\ell,blue,blue) + I(\ell,blue,red)\Big) + I(j,blue) \text{, \quad and}$$
		$$\text{cnt}(i,j,red) = I(i,red) + \sum_{\ell=i}^{j-1}\Big(I(\ell,red,red) + I(\ell,blue,red)\Big) \text{.}$$		

		In case $C_i \neq C_j$, we have $I(j,blue)=I(i,red)$, and
		$$\text{cnt}(i,j,blue) - \text{cnt}(i,j,red) = \sum_{\ell=i}^{j-1}I(\ell,blue,blue) - \sum_{\ell=i}^{j-1}I(\ell,red,red) \text{,}$$
		that is, 
		the difference between the number of blue points and the number of red points equals the difference between the number of blue-blue pairs and and the number of red-red pairs.

		Looking at the indices of the consecutive vertices in $Q$, at each blue-blue pair the vertex index increases by one, and at each red-red pair the vertex index decreases by one. Further on, at blue-red and red-blue pairs the index of the consecutive vertices stays the same. Hence, if $C_i \neq C_j$, we have that $Q_i = Q_j$ if and only if the numbers of blue and red points in $\p{i}{j}$ are equal, which is, by the definition of an orbit, equivalent to $\OO{i} = \OO{j}$.
		
		Moving on to the case $C_i = C_j$, we observe that there is always a $j'$ such that $Q_{j'} = Q_j$ and $C_{j'} \neq C_j$. Our analysis from the previous case readily gives $\OO{j} = \OO{j'}$. Applying the same argument again, now for $i$ and $j'$, we get that $Q_i = Q_{j'}$ if and only if $\OO{i} = \OO{j'}$. This directly implies that $Q_i = Q_{j}$ if and only if $\OO{i} = \OO{j}$.

		Combining the two cases, we have $\OO{i} = \OO{j}$ if and only if $Q_i = Q_{j}$, for all values of $i$ and $j$.

		As each vertex of $G$ participates in $Q$, we established a bijective correspondence between the vertices of $G$ and the orbits of $P$. It is left to show that two orbits of $P$ cross if and only if their corresponding vertices of $G$ are connected with an arc in $G$, which, because $G$ is segmented, is the same as the condition that both vertices lie under the same segment. More precisely, for each $i$ and $j$, where $i < j$, we want to show that orbits corresponding to $v_i$ and $v_j$ cross if and only if there is $k \in \{0, \ldots, Z-1\}$ such that $s_k \leq i < j \leq t_k$.
		
		It is easy to see that two orbits with corresponding vertices $v_i$ and $v_j$ cross if and only if $v_i v_j v_i v_j$ appears as a cyclic subsequence (not necessarily consecutive) of $Q$.
		
		If there is $k \in \{0, \ldots, Z-1\}$ such that $s_k \leq i < j \leq t_k$, then \\
		$v_i \in (s_{k-1}{\nearrow}{t_k})$,\\
		$v_j \in (t_{k}{\searrow}{s_k})$,\\
		$v_i \in (s_{k}{\nearrow}{t_{k+1}})$, and\\
		$v_j \in (t_{Z}{\searrow}{s_Z}) = ((m-1){\searrow}0)$.
		
		All of the subsequences on the right sides are disjoint subsequences of consecutive elements of $Q$, in order as they appear in $Q$, so $v_i v_j v_i v_j$ is indeed a (not necessarily consecutive) subsequence of $Q$, and their corresponding orbits cross.
		
		On the other hand, consider the case when there is no $k \in \{0, \ldots, Z-1\}$ such that $s_k \leq i < j \leq t_k$. Let $\ell$ be the largest index such that $i \in \{s_\ell, \ldots, t_\ell\}$. All appearances of $v_i$ in $Q$ except the last one are in
		$\Big( \bigoplus_{e=0}^{\ell} \big( (s_{e-1}{\nearrow}{t_e}) \oplus (t_{e}{\searrow}{s_e}) \big) \Big) \oplus (s_{\ell}{\nearrow}{t_{\ell+1}})$,
		and all appearances of $v_j$ in $Q$ except the last one are in
		$\Big( \bigoplus_{e=\ell+1}^{Z-1} \big( (s_{e-1}{\nearrow}{t_e}) \oplus (t_{e}{\searrow}{s_e}) \big) \Big) \oplus (s_{Z-1}{\nearrow}{t_Z})$.
		Notice that if there are both $v_i$ and $v_j$ in $(s_{\ell}{\nearrow}{t_{\ell+1}})$, then the appearance of $v_i$ in that part comes before the appearance of $v_j$ in that part, since $i < j$ and the part is increasing. So, all the non-last appearances of $v_i$ come before all the non-last appearances of $v_j$. There is one additional appearance of both $v_i$ and $v_j$ in the last part, $(t_{Z}{\searrow}{s_Z}) = ((m-1){\searrow}0)$, but there $v_j$ comes before $v_i$, since that part is decreasing. Therefore, $v_i v_j v_i v_j$ is not a cyclic subsequence of $Q$, meaning their corresponding orbits do not cross.
	\end{proof}

	The last proposition we state is a consequence of the fact that orbit graphs are segmented. A \emph{weakly connected component} of a directed graph is a connected component of the undirected graph obtained from the directed graph by removing the edge directions.
	
	\begin{proposition}
		\label{prp:HamiltonianPath}
		Each weakly connected component of $\GR(P)$ contains a unique Hamiltonian path.
	\end{proposition}
	\begin{proof}	
		Let $\OL$ and $\OH$ respectively be the lowest and the highest orbit of some weakly connected component of $\GR(P)$, relative to the total order of orbits. By definition, there is an undirected path between $\OL$ and $\OH$. For any two consecutive orbits $\OB$ and $\OC$ in the total order such that $\OL \preceq \OB \preceq \OC \preceq \OH$, there is an arc $(\OA,\OD)$ corresponding to an edge in some undirected path between $\OL$ and $\OH$, so that $(\OA,\OD)$ goes ``over'' $\OB$ and $\OC$. More formally, there are $\OA$ and $\OD$ such that $\OL \preceq \OA \preceq \OB \preceq \OC \preceq \OD \preceq \OH$, and $(\OA,\OD) \in E(\GR(P))$. By a direct application of Proposition~\ref{prp:segmentedStructure} we get that $(\OB, \OC) \in E(\GR(P))$.
		
		Therefore, the sequence of all orbits from $\OL$ to $\OH$, ordered by the relation $\preceq$, makes a Hamiltonian path in this weakly connected component. It is a unique Hamiltonian path since $\GR(P)$ is a subgraph of a total order.
	\end{proof}

	Finally, we show how to efficiently compute the total order and Hamiltonian paths. It will later be needed for constructing efficient algorithms dealing with non-crossing matchings.

	\begin{lemma}
		\label{lem:OrbitGraphComplexity}
		The total order of orbits, and the Hamiltonian paths for all weakly connected components of the orbit graph can be found in $O(n)$ time in total.
	\end{lemma}
	\begin{proof}
		Our goal here is to compute $succ(\OA)$ and $succG(\OA)$ for each orbit $\OA$, defined as the successor of $\OA$ in the total order of orbits, and the successor of $\OA$ in the corresponding Hamiltonian path, respectively. (Undefined values of these functions mean that there is no successor in the respective sequence.) Having these two functions computed, it is then easy to reconstruct the total order and the Hamiltonian paths. We start by computing the orbits in $O(n)$ time, as described in Lemma~\ref{lem:OrbitsComplexity}.
		
		By Proposition~\ref{prp:Consecutive}, for every two consecutive orbits in the total order, there are at least two consecutive points on $P$, one from each of those orbits. We scan through all consecutive pairs of points on $P$. Let $i$ and $i+1$ be two consecutive points. If they have different color, then they belong to the same orbit and we do nothing in this case. If their color is the same, they belong to different orbits, and from Proposition~\ref{prp:Consecutive} we know that those two orbits are consecutive in the total order. If the color of the points is blue, then there is a point $i+1$ from $\OO{i+1}$ on the right side of blue-red edge $(i,o(i))$ from $\OO{i}$, so we conclude that $\OO{i} \preceq \OO{i+1}$, and we set $succ(\OO{i}) = \OO{i+1}$. In the other case, when the points are red, we set $succ(\OO{i+1}) = \OO{i}$. It only remains to check whether these two orbits cross. If they cross anywhere, then edges $(i,o(i))$ and $(o^{-1}(i+1),i+1)$ must cross each other (otherwise, the whole $\OO{i+1}$ would lie on the right side of $(i,o(i))$), so it is enough to check only for this pair of edges whether they cross. If they do cross, we do the same with the function $succG$, that is, we either set $succG(\OO{i}) = \OO{i+1}$ if the points are blue, or $succG(\OO{i+1}) = \OO{i}$ if they are red. If they do not cross, we do not do anything.
		
		Constructing the corresponding sequences of orbits is done by first finding the orbits which are not successor of any other orbit and then just following the corresponding successor function.
		
		The whole process takes $O(n)$ time in total.
	\end{proof}

	\section{Finding bottleneck matchings}
	\label{sec:BottleneckMatchingsB}
	
	For the problem of finding a bottleneck bichromatic matching of points in convex position, we will utilize the theory that is developed for orbits and the orbit graph, combining it with the approach used in \cite{savic2017faster} to tackle the monochromatic case.
	
	For the special configuration where colors alternate, i.e.,~two points are colored the same if and only if the parity of their indices is the same, we note that every pair $(i,j)$ where $i$ and $j$ are of different parity is feasible. This is also the case with the monochromatic version of the same problem, so since the set of pairs that is allowed to be matched is the same in both cases, the bichromatic problem is in a way a generalization of the monochromatic problem -- to solve the monochromatic problem it is enough to color the points in an alternating fashion, and then apply the algorithm which solves the bichromatic problem.
	
	
	
	\begin{definition}[Turning angle, $\tau$]
		\fig{fig:TurningAngleB}{Turning angle.}{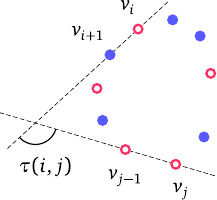}%
		The \emph{turning angle} of $\pij$, denoted by $\tau(i,j)$, is the angle by which the vector $\overrightarrow{v_iv_{i+1}}$ should be rotated in the positive direction to align with the vector $\overrightarrow{v_{j-1}v_j}$, see Figure~\ref{fig:TurningAngleB}.
	\end{definition}
	
	\begin{lemma}
		\label{lem:PiHalfB}
		There is a bottleneck matching $M$ of $P$ such that all diagonals $(i,j) \in M$ have $\tau(i,j) > \pi/2$.
	\end{lemma}
	
	To prove this lemma, we use the same approach as in~\cite[Lemma 1]{savic2017faster}. The proof is deferred to Appendix.
	
	Next, we consider the division of the interior of the polygon $\mathcal{P}$ into regions obtained by cutting it along all diagonals (but not edges) from the given matching $M$. Each region created by this division is bounded by some diagonals of $M$ and by the boundary of the polygon $\mathcal{P}$.
	
	\begin{definition}[Cascade, $k$-bounded region]
		\fig{fig:CascadesB}{Matching consisting of edges (dashed lines) and diagonals (solid lines). Orbits are denoted by gray shading. \\ There are three cascades in this example: one consist of the three diagonals in the upper part, one consist of the two diagonals in the lower left, and one consist of the single diagonal in the lower right.}{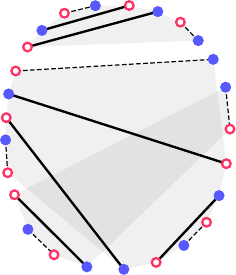}%
		Regions bounded by exactly $k$ diagonals are called \emph{$k$-bounded} regions. Any maximal sequence of diagonals connected by $2$-bounded regions is called a \emph{cascade} (see Figure~\ref{fig:CascadesB} for an example).
	\end{definition}

	\begin{lemma}
		\label{lem:ThreeCascadesB}
		There is a bottleneck matching having at most three cascades.
	\end{lemma}
	To prove this lemma, we use the same approach as in~\cite[Lemma 2]{savic2017faster}. The proof is deferred to Appendix.
	
	
	
	It is not possible for a matching to have exactly two cascades. If there were exactly two cascades, there would be a region defined by diagonals from both cascades. If that region were bounded by exactly one diagonal from each cascade, it would then be $2$-bounded and, by definition of cascade, those two diagonals would belong to the same cascade. Otherwise, if that region were bounded by more than one diagonal from one of the two cascades, it would then be at least $3$-bounded and, by definition of cascade, no two of its diagonals would belong to the same cascade, and hence we would have more than two cascades.
	
	So, from Lemma~\ref{lem:ThreeCascadesB} we know that there is a bottleneck matching which either has at most one cascade and no $3$-bounded regions, or it has a single $3$-bounded region and exactly three cascades. In the following section we define a set of more elementary problems that will be used to find an optimal solution in both of these cases.

	\subsection{Matchings with at most one cascade}
	\label{sec:SubproblemsB}
	
	When talking about matchings with minimal value under certain constraints, we will refer to these matchings as \emph{optimal}.
	
	\begin{definition}[$\Matching^0$, $M^0$]
		For $i$ and $j$ such that $\pij$ is balanced, let $\Matching^0(i,j)$ be the problem of finding an optimal matching $M^0_{i,j}$ of points in $\pij$ using edges only.
	\end{definition}
	
	\begin{definition}[$\Matching^1$, $M^1$]
		For $i$ and $j$ such that $\pij$ is balanced, let $\Matching^1(i,j)$ be the problem of finding an optimal matching $M^1_{i,j}$ of points in $\pij$, so that $M^1_{i,j}$ has at most one cascade, and the segment $(i,j)$ belongs to a region bounded by at most one diagonal from $M^1_{i,j}$ different from $(i,j)$.
	\end{definition}
	
	When $\pij$ is balanced, Proposition~\ref{prp:MatchingWithEdgesBalanced} ensures that solutions to $\Matching^0(i,j)$ and $\Matching^1(i,j)$ always exist, that is, $M^0_{i,j}$ and $M^1_{i,j}$ are well defined.

	Let $i$ and $j$ be such that $\pij$ is balanced. First, let us analyze how $\Matching^0(i,j)$ can be reduced to smaller subproblems. The point $i$ can be matched either with $o(i)$ or with $o^{-1}(i)$. The first option is always possible because Proposition~\ref{prp:BalancedO} states that $o(i) \in \pij$, but the second one is possible only if $o^{-1}(i) \in \pij$ (it is also possible that $o(i) = o^{-1}(i)$, but no special analysis is needed for that). In the first case, $M^0(i,j)$ is constructed as the union of $(i,o(i))$, and optimal edge-only matchings for point sets $\p{i+1}{o(i)-1}$, if $|\p{i}{o(i)}| > 2$, and $\p{o(i)+1}{j}$, if $o(i) \neq j$, since both sets are balanced. The second case is similar, $M^0(i,j)$ is constructed as the union of $(o^{-1}(i),i)$, and optimal edge-only matchings for point sets $\p{i+1}{o^{-1}(i)-1}$, if $|\p{i}{o^{-1}(i)}| > 2$, and $\p{o^{-1}(i)+1}{j}$, if $o^{-1}(i) \neq j$, since both sets are balanced.
	
	Next, we show how to reduce $\Matching^1(i,j)$ to smaller subproblems. If $i$ and $j$ have different colors, then $(i,j)$ is a feasible pair, and it is possible that $M^1_{i,j}$ includes this pair. In that case, $M^1_{i,j}$ is obtained by taking $(i,j)$ together with $M^1(i+1,j-1)$, if $|\pij| > 2$, since $\p{i+1}{j-1}$ is balanced. Now, assume that $i$ is not matched to $j$ (no matter whether $(i,j)$ is feasible or not). Let $k$ and $l$ be the points in $\pij$ which are matched to $i$ and $j$ in the matching $M^1_{i,j}$, respectively. By the requirement, $(i,k)$ and $(l,j)$ cannot both be diagonals, otherwise $(i,j)$ would belong to the region bounded by more than one diagonal from $M^1_{i,j}$. If $(i,k)$ is an edge, then, depending on the position of the diagonals that belong to the single cascade of $M^1_{i,j}$, the matching is constructed by taking $(i,k)$ together either with $M^0_{i+1,k-1}$, if $|\p{i}{k}| > 2$, and $M^1_{k+1,j}$, if $k \neq j$, or with $M^1_{i+1,k-1}$, if $|\p{i}{k}| > 2$, and $M^0_{k+1,j}$, if $k \neq j$. Similarly, if $(l,j)$ is an edge, then $M^1_{i,j}$ is constructed by taking $(l,j)$ together either with $M^0_{l+1,j-1}$, if $|\p{l}{j}| > 2$, and $M^1_{i,l-1}$, if $i \neq l$, or with $M^1_{l+1,j-1}$, if $|\p{l}{j}| > 2$, and $M^0_{i,l-1}$, if $i \neq l$. All the mentioned matchings exist because their respective underlying point sets are balanced.
	
	As these problems have optimal substructure, we can apply dynamic programming to solve them. If $\bn(M^0_{i,j})$ and $\bn(M^1_{i,j})$ are saved into $S^0(i,j)$ and $S^1(i,j)$, respectively, the following recursion formulas can be used to compute the solutions to $\Matching^0(i,j)$ and $\Matching^1(i,j)$ for all pairs $(i,j)$ such that $\pij$ is balanced.
	\[
	S^0(i,j) = \min
	\begin{cases}
	\max
	\begin{cases}
	& |v_iv_{o(i)}| \\
	\text{if $|\p{i}{o(i)}| > 2$ :} & S^0(i+1,o(i)-1) \\
	\text{if $o(i) \neq j$ :}       & S^0(o(i)+1,j) \\
	\end{cases}\\
	\text{if $(o^{-1}(i) \in \pij)$ :}\\
	\qquad\max
	\begin{cases}
	& |v_iv_{o^{-1}(i)}| \\
	\text{if $|\p{i}{o^{-1}(i)}| > 2$ :} & S^0(i+1,o^{-1}(i)-1) \\
	\text{if $o^{-1}(i) \neq j$ :}       & S^0(o^{-1}(i)+1,j) \\
	\end{cases}\\
	\end{cases}
	\]
	\[
	S^1(i,j) = \min
	\begin{cases}
	\max
	\begin{cases}
	& |v_iv_{o(i)}| \\
	\text{if $|\p{i}{o(i)}| > 2$ :} & S^0(i+1,o(i)-1) \\
	\text{if $o(i) \neq j$ :}       & S^1(o(i)+1,j) \\
	\end{cases} \\
	\max
	\begin{cases}
	& |v_iv_{o(i)}| \\
	\text{if $|\p{i}{o(i)}| > 2$ :} & S^1(i+1,o(i)-1) \\
	\text{if $o(i) \neq j$ :}       & S^0(o(i)+1,j) \\
	\end{cases} \\			
	\max
	\begin{cases}
	& |v_{o^{-1}(j)}v_j| \\
	\text{if $|\p{o^{-1}(j)}{j}| > 2$ :} & S^0(o^{-1}(j)+1,j-1) \\
	\text{if $o^{-1}(j) \neq i$ :}       & S^1(i,o^{-1}(j)-1) \\
	\end{cases} \\
	\max
	\begin{cases}
	& |v_{o^{-1}(j)}v_j| \\
	\text{if $|\p{o^{-1}(j)}{j}| > 2$ :} & S^1(o^{-1}(j)+1,j-1) \\
	\text{if $o^{-1}(j) \neq i$ :}       & S^0(i,o^{-1}(j)-1) \\
	\end{cases} \\
	\text{if $(i,j)$ is feasible:}\\
	\qquad\max
	\begin{cases}
	& |v_iv_j| \\
	\text{if $\pij > 2$ :} & S^1(i+1,j-1) \\
	\end{cases} \\
	\end{cases}
	\]
	
	We fill values of $S^0$ and $S^1$ in order of increasing $j-i$, so that all subproblems are already solved when needed.
	
	Beside the value of a solution $\Matching^1(i,j)$, it is going to be useful to determine if pair $(i,j)$ is necessary for constructing $M^1_{i,j}$.
	
	\begin{definition}[Necessary pair]
		We call an oriented pair $(i,j)$ \emph{necessary} if it is contained in every solution to $\Matching^1(i,j)$.
	\end{definition}
	
	Obviously, a pair can be necessary only if it is feasible. Computing whether $(i,j)$ is a necessary pair can be easily incorporated into the computation of $S^1(i,j)$. Namely the pair $(i,j)$ is necessary, if $(i,j)$ is an edge, or the equation for $S^1(i,j)$ achieves the minimum only in the last case (when $(i,j)$ is feasible). If it is necessary, we set $necessary(i,j)$ to true, otherwise we set it to false. Note that $necessary(i,j)$ does not imply $necessary(j,i)$.
	
	We have the total of $O(n^2)$ subproblems each of which takes $O(1)$ time to be computed, assuming that $o$ and $o^{-1}$ have already been computed. Hence, all computations together require $O(n^2)$ time and the same amount of space.
	
	Note that we computed only the values of solutions to all subproblems. If an actual matching is needed, it can be easily reconstructed from the data in $S$ in linear time per subproblem.
	
	We note that every matching with at most one cascade has a feasible pair $(k,k+1)$ such that the segment $(k,k+1)$ belongs to a region bounded by at most one diagonal from that matching. Indeed, if there are no diagonals in the matching, any pair $(k,k+1)$ where $k$ and $k+1$ have different colors satisfies the condition. If there is a cascade, we take one of the two endmost diagonals of the cascade, let it be $(i,j)$, so that there are no other diagonals from $M$ in $\pij$. Since $\pij$ is balanced, there are two neighboring points $k, k+1 \in \pij$ with different colors, and the pair $(k,k+1)$ is the one we are looking for.
	
	Now, an optimal matching with at most one cascade can be found easily from precomputed solutions to subproblems by finding the minimum of all $S^1(k+1,k)$ for all feasible pairs $(k,k+1)$ and reconstructing $M^1_{k+1,k}$ for $k$ that achieved the minimum. The last (reconstruction) step takes only linear time.

	\subsection{Matchings with three cascades}
	\label{sec:MatchingsWithThreeCascadesB}
	
	As we already concluded, there is a bottleneck matching of $P$ having either at most one cascade, or exactly three cascades. An optimal matching with at most one cascade can be found easily from computed solutions to subproblems, as shown in the previous section. We now focus on finding an optimal matching among all matchings with exactly three cascades, denoted by \emph{$3$-cascade matchings} in the following text.

	Any three distinct points $i$, $j$ and $k$ with $j \in \p{i+1}{k-1}$, where $(i,j)$, $(j+1, k)$ and $(k+1, i-1)$ are feasible pairs, can be used to construct a $3$-cascade matching by simply taking a union of $M^1_{i,j}$, $M^1_{j+1,k}$ and $M^1_{k+1,i-1}$. (Note that these three feasible pairs do not necessarily belong to the combined matching, since they might not be necessary pairs in their respective $1$-cascade matchings.)
	
	To find the optimal matching we could run through all possible triplets $(i,j,k)$ such that $(i,j)$, $(j+1, k)$ and $(k+1, i-1)$ are feasible pairs, and see which one minimizes $\max\{S^1[i,j], S^1[j+1,k], S^1[k+1,i-1]\}$. However, this requires $O(n^3)$ time, and thus is not suitable, since our goal is to design a faster algorithm. Our approach is to show that instead of looking at all $(i,j)$ pairs, it is enough to select $(i,j)$ from a set of linear size, which would reduce the search space to quadratic number of possibilities, so the search would take only $O(n^2)$ time.


	\begin{definition}[Inner diagonals, Inner region, Inner pairs]
		In a $3$-cascade matching, we call the three diagonals at the inner ends of the three cascades the \emph{inner diagonals}. We take the largest region by area, such that it is bounded but not crossed by the segments of the matching, and such that every pair of the three cascades is separated by that region, and we call this region the \emph{inner region}. Pairs of points matched by segments that are on the boundary of the inner region are called the \emph{inner pairs}. For an example, see Figure~\ref{fig:InnerPairs}.
		
		\fig{fig:InnerPairs}{
			The edges are represented by dashed lines, and the diagonals with solid lines. There are three cascades, $\{\{(11,18),(14,17)\},\{(22,27),(23,26)\},\{(30,9),(1,6),(2,5)\}\}$.
			The inner diagonals are $\{(11,18),(22,27),(30,9)\}$, and the
			inner pairs are $\{(10,19),(20,21),(22,27),(28,29),(30,9)\}$.%
		}{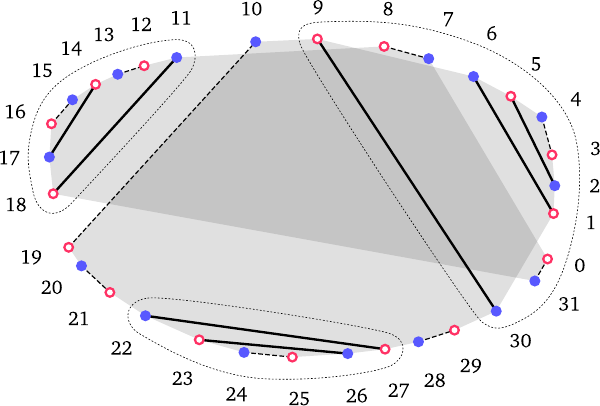}
	\end{definition}
	
	Since the inner region separates the cascades, there must be at least $3$ inner pairs.
	
	\begin{lemma}
		\label{lem:BottleneckWithAllNecessaryB}
		If there is no bottleneck matching with at most one cascade, then there is a bottleneck $3$-cascade matching whose every inner pair is necessary.
	\end{lemma}
	To prove this lemma, we use the same approach as in~\cite[Lemma 3]{savic2017faster}. The proof is deferred to Appendix.
	
	\begin{definition}[Candidate pair, Candidate diagonal]
		An oriented pair $(i,j)$ is a \emph{candidate pair}, if it is a necessary pair and $\tau(i,j) \leq 2\pi/3$. If a candidate pair is a diagonal, it is called a \emph{candidate diagonal}.
	\end{definition}
	
	\begin{lemma}
		\label{lem:BottleneckWithCandidateB}
		If there is no bottleneck matching with at most one cascade, then there is a $3$-cascade bottleneck matching $M$, such that at least one inner pair of $M$ is a candidate pair.
	\end{lemma}
	To prove this lemma, we use the same approach as in~\cite[Lemma 4]{savic2017faster}. The proof is deferred to Appendix.

	\fig{fig:PolarityDefinitionsB}{Geometric regions used for locating points $v_{i+1}, \ldots, v_{j-1}$.}{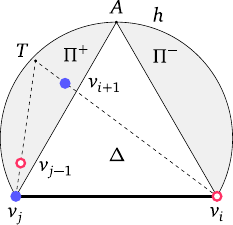}%
	Let us now take a look at an arbitrary candidate diagonal $(i,j)$, and examine the position of points $\p{i}{j} \cap \OO{i}$ relative to it. To do that, we locate points $v_i$ and $v_j$ and then define several geometric regions relative to their position, inspired by the geometric structure used in~\cite{savic2017faster} to tackle the monochromatic version of the problem.
	
	Firstly, we construct the circular arc $h$ on the right side of the directed line $v_iv_j$, from which the line segment $v_iv_j$ subtends an angle of $\pi/3$, see Figure~\ref{fig:PolarityDefinitionsB}. Let $A$ be the midpoint of $h$. Points $v_i$, $A$ and $v_j$ form an equilateral triangle which we denote by $\Delta$. Let $\Pi^-$ be the region bounded by $h$ and the line segment $v_iA$, and $\Pi^+$ the region bounded by $h$ and the line segment $Av_j$.
	
	The following lemma is crucial in our analysis of bichromatic bottleneck matchings. Even though in statement it is similar to~\cite[Lemma 5]{savic2017faster}, which was developed to tackle monochromatic bottleneck matchings, the proof we show here is much more involved, capturing the specifics of the bichromatic version of the problem and making use of the theory we developed around orbits.
	
	\begin{lemma}
		\label{lem:PolarityB}
		For every candidate diagonal $(i,j)$, the points from $\pij \cap \OO{i}$ other than $i$ and $j$ lie either all in $\Pi^-$ or all in $\Pi^+$.
	\end{lemma}
	\begin{proof}
		W.l.o.g.~let us assume that point $i$ is red. Since $(i,j)$ is a diagonal, there are more than two points in $\pij \cap \OO{i}$. Let $T$ be the point of intersection of lines $v_iv_{i+1}$ and $v_jv_{j-1}$, see Figure~\ref{fig:PolarityDefinitionsB}. Since $\tau(i,j) \leq 2\pi/3$, the point $T$ lies in the region bounded by the line segment $v_iv_j$ and the arc $h$. Because of convexity, all points in $\pij$ must lie inside the triangle $\triangle{v_iTv_j}$, so there cannot be two points from $\pij$ such that one is on the right side of the directed line $v_iA$ and the other is on the left side of the directed line $v_jA$. This implies that either $\Pi^-$ or $\Pi^+$ is empty.
		
		W.l.o.g., let us assume that there are no points from $\pij$ on the right side of the directed line $v_iA$, so all points in $\pij$ lie in $\Pi^+ \cup \Delta$. It is important to note that the diameter of both $\Pi^+$ and $\Delta$ is $|v_iv_j|$, that is, no two points both inside $\Pi^+$ or both inside $\Delta$ are at a distance of more than $|v_iv_j|$.
		
		To complete the proof, we need to prove that no points of $\pij \cap \OO{i}$ other than $i$ and $j$ lie in $\Delta$, so for a contradiction we suppose the opposite, that there is at least one such point in $\Delta$.
		
		We denote the set of points in $\Pi^+$ (including $j$) with $U$. If there are points on $Av_j$, we consider them to belong to $U$. The pair $(i,j)$ is a feasible pair, so, by Proposition~\ref{prp:FeasibleSplitsBalanced}, the number of points from any orbit inside $\pij$ is even, implying that the parity of $|U \cap \OO{i}|$ is the same as the parity of $|(\pij \setminus U) \cap \OO{i}|$. We will analyze two cases depending on the parity of the number of points in $U \cap \OO{i}$.
		
		\paragraph{Case 1}
		There is an even number of points in $U \cap \OO{i}$, and thus also in $(\pij \setminus U) \cap \OO{i}$.
		
		Let $M$ be an optimal matching of points in $\pij$. The pair $(i,j)$ is a candidate pair, and thus necessary, so it is contained in every optimal matching of points in $\pij$, including $M$, and hence $\bn(M) \geq |v_iv_j|$. To complete the proof in this case, we will construct another optimal matching $M'$ that does not contain the pair $(i,j)$, by joining two newly constructed matchings, $M'_{out}$ and $M'_{in}$, thus arriving to a contradiction to the assumption that the pair $(i,j)$ is a candidate pair.
		
		\fig{fig:PolarityProofCase1}{$M'_{in}$ and $M'_{out}$; only points from $\OO{i}$ are depicted as points.}{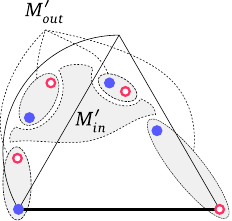}%
		We obtain the matching $M'_{out}$ by arbitrarily matching the set $\p{l}{o(l)}$, for each red-blue edge $(l,o(l))$ of $\OO{i}$ in $\pij$, as illustrated in Figure~\ref{fig:PolarityProofCase1} (note that in the figure only points from $\OO{i}$ are depicted as points). More formally, $M'_{out}$ is a union of matchings of sets $\p{o^{2k}(i)}{o^{2k+1}(i)}$, for each $k \in \{0,1,\ldots,(s-1)/2\}$, where $s$ is the smallest positive integer such that $o^s(i) = j$ (by Proposition~\ref{prp:FeasibleSplitsBalanced} and Lemma~\ref{lem:MatchingAlwaysPossible}, all these matchings exists). Since $|U \cap \OO{i}|$ is even, points of each pair in $M'_{out}$ are either both in $U$ or both in $\pij \setminus U$, that is, they are either both in $\Pi^+$ or both in $\Delta$, so the distance of each pair is at most $|v_iv_j|$, implying $\bn(M'_{out})\leq |v_iv_j|$.
		
		The rest of the points in $\pij$ are all on the right side of blue-red edges of $\OO{i}$, and by Proposition~\ref{prp:OrbitInteractionParity} the points they are paired up with in $M$ are also on the right side of blue-red edges of $\OO{i}$. Therefore, all those pairs are unobstructed by the segments in $M'_{out}$, and we can simply define $M'_{in}$ to be the restriction of $M$ to the set of those points from $\pij$ that are on the right side of blue-red edges of $\OO{i}$.

		
		All points in $\pij$ are covered by $M' = M'_{out} \cup M'_{in}$, and we have that $\bn(M') = \max\{\bn(M'_{in}), \bn(M'_{out})\} \leq \max\{\bn(M),|v_iv_j|\} = \bn(M)$. Since $M$ is optimal, the equality holds and $M'$ is optimal too. We constructed an optimal matching $M'$ on $\pij$ that does not contain the pair $(i,j)$, but such a matching cannot exist because $(i,j)$ is a candidate diagonal, and hence necessary, a contradiction.

		\paragraph{Case 2}
		\fig{fig:PolarityProofCase2UVW}{$U$, $V$ and $W$; only points from $\OO{i}$ are depicted as points.}{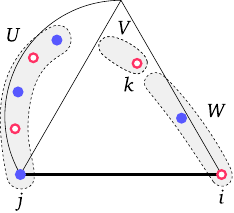}%
		There is an odd number of points in $U \cap \OO{i}$, and thus also in $(\pij \setminus U) \cap \OO{i}$.
		
		Let $k$ be the last point from the set $\pij \cap \OO{i}$ (observed in the positive direction around $\mathcal{P}$, starting from $i$) that lies in $\Delta$, see Figure~\ref{fig:PolarityProofCase2UVW}. Note that $k$ must have the same color as $i$. We define $V := \p{k}{j} \setminus U$ and $W := \p{i}{k-1}$ (we earlier assumed that there is at least one point from $\OO{i}$ other than $i$ in $\Delta$, so $k \neq i$).
		
		By $M$ we denote an optimal matching of points in $\pij$ that minimizes the number of matched pairs between $U$ and $W$. The pair $(i,j)$ is a candidate pair, so it is a necessary pair, that is, every optimal matching of points in $\pij$ contains $(i,j)$, meaning that there is at least one matched pair between $U$ and $W$ in $M$. Let $a$ be the last point in $W$ (observed in the positive direction around $\mathcal{P}$, starting from $i$) matched to a point in $U$, and let $b$ be the point from $U$ it is matched to, i.e.,~$(a,b) \in M$.
		
		\fig{fig:PolarityProofCase2abef}{Even number of points in $\p{i}{a} \cap \OO{a}$.}{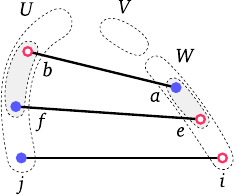}%
		If there is an even number of points in $\p{i}{a} \cap \OO{a}$, then the numbers of red and blue points in that set are equal, so at least one of those points (which has a different color from $a$) must be matched to a point in $U$ as well. Let that point be $e$ and let its pair in $U$ be $f$, see Figure~\ref{fig:PolarityProofCase2abef}.%
		
		We can now modify the matching by replacing $(a,b)$, $(e,f)$, and all the matched pairs between them with a matching of points in $\p{e}{a}$, and a matching of points in $\p{b}{f}$, which is possible by Proposition~\ref{prp:FeasibleSplitsBalanced} and Lemma~\ref{lem:MatchingAlwaysPossible}. Each newly matched pair has both its endpoints in the same set, either $U$ or $W$, so its distance is at most $|v_iv_j|$, meaning that this newly constructed matching is optimal as well. This, however, reduces the number of matched pairs between $U$ and $W$ while keeping the matching optimal, which contradicts the choice of $M$, so there must be an odd number of points in $\p{i}{a} \cap \OO{a}$.

		\fig{fig:PolarityProofCase2abcd}{Odd number of points in $\p{i}{a} \cap \OO{a}$.}{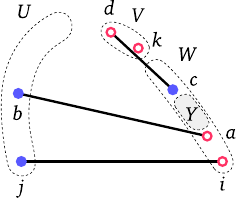}%
		As the number of points from $\OO{i}$ in $V \cup W$ is odd, and the only point in $V$ from $\OO{i}$ is $k$, there is an even number of points from $\OO{i}$ in $W$. Since $i$ and $k$ belong to the same orbit, there is an even number of points from any particular orbit in $W$ (as a consequence of applying Proposition~\ref{prp:FeasibleSplitsBalanced} to each pair of consecutive points of $\OO{i}$ inside $W$). As there is an odd number of points in $\p{i}{a} \cap \OO{a}$, there is an even number of points in $\p{a}{k-1} \cap \OO{a}$, so at least one of them with a color different from $a$ must be matched with a point outside of $W$. Let $c$ be the first such point in $\p{a}{k-1}$ (observed in the positive direction around $\mathcal P$, starting from $a$), see Figure~\ref{fig:PolarityProofCase2abcd}. The way we chose $a$ implies that $c$ cannot be matched to some point in $U$, so it must be matched to a point in $V$, let us call it $d$.
		
		Let us denote the set $\p{a}{c} \setminus \{a,c\}$ by $Y$. The choice of $a$ guarantees that no point in $Y$ is matched to a point in $U$. Points $a$ and $c$ belong to the same orbit, so by Proposition~\ref{prp:FeasibleSplitsBalanced} there is an even number of points from any particular orbit in $Y$. Hence, if there is a point $g_1$ in $Y$ matched to a point $h_1$ in $V$, then there must be another matched pair $(g_2,h_2)$ from the same orbit such that $g_2 \in Y$, $h_2 \in V$, and $g_1$ and $g_2$ have different colors. We modify the matching by replacing $(g_1,h_1)$, $(g_2,h_2)$ and all the matched pairs between them with a matching $M_g$ of points in $\p{g_1}{g_2}$, and a matching $M_h$ of points in $\p{h_1}{h_2}$. This is again possible by Proposition~\ref{prp:FeasibleSplitsBalanced} and Lemma~\ref{lem:MatchingAlwaysPossible}. Matchings $M_g$ and $M_h$ are fully contained in $W$ and $V$, respectively, so no matched pair of theirs is at a distance greater than $|v_iv_j|$, and the newly obtained matching is optimal as well. By iteratively applying this modification we can eliminate all matched pairs between $Y$ and $V$, so that finally there is no matched pairs going out from $Y$, meaning no matched pair crosses either $(a,c)$ or $(b,d)$.
		
		We are now free to ``swap'' the matched pairs between points $a$, $b$, $c$, and $d$, by replacing $(a,b)$ and $(c,d)$ with $(a,c)$ and $(b,d)$, because no other matched pair can possibly cross the newly formed pairs. We need to show that this swap does not increase the value of the matching. The pair $(a,c)$ cannot increase the matching value because $a$ and $c$ are both in $W$, so their distance is at most $|v_iv_j|$. To show that the pair $(b,d)$ also does not increase the value of the matching, we consider two cases based on the position of the point $d$.
		
		\fig{fig:PolarityProofCase2bdShort}{$d$ lies in $\triangle v_iZv_j$.}{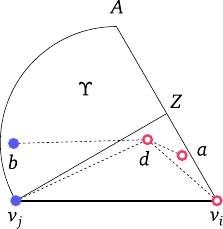}%
		Let $Z$ be the midpoint of the line segment $v_iA$. Let us denote the region $(\Pi^+ \cup \Delta) \setminus \triangle v_iZv_j$ by $\Upsilon$. No two points in $\Upsilon$ are at a distance greater than $|v_iv_j|$. The point $b$ lies in $\Upsilon$. If the point $d$ lies in $\Upsilon$ as well, then $|bd| \leq |v_iv_j|$. Otherwise, $d$ lies in $\triangle v_iZv_j$, see Figure~\ref{fig:PolarityProofCase2bdShort}, and $\measuredangle adb > \measuredangle v_idv_j > \measuredangle v_iZv_j = \pi/2$ (the first inequality holds because the points are in convex position). The angle $\measuredangle adb$ is hence obtuse, and therefore $|bd| < |ab|$. But the pair $(a,b)$ belongs to the original matching $M$, so the newly matched pair $(b,d)$ also does not increase the value of the matching.
		
		By making modifications to the matching $M$ we constructed a new matching $M'$ with the value not greater than the value of $M$. Since $M$ is optimal, these values are actually equal, and the matching $M'$ is also optimal. However, the pair $(a,b)$ is contained in $M$, but not in $M'$, and we did not introduce new matched pairs between $U$ and $W$, so there is a strictly smaller number of matched pairs between $U$ and $W$ in $M'$ than in $M$, which contradicts the choice of $M$.
		
		The analysis of both Case 1 and Case 2 ended with a contradiction, which completes the proof of the lemma.
	\end{proof}

	With $\Pi^-(i,j)$ and $\Pi^+(i,j)$ we respectively denote regions $\Pi^-$ and $\Pi^+$ corresponding to an ordered pair $(i,j)$. For candidate diagonals, the existance of the two possibilities given by Lemma~\ref{lem:PolarityB} induces a concept of \emph{polarity}.
	
	\begin{definition}[Polarity, Pole]
		Let an oriented pair $(i,j)$ be a candidate diagonal. If all points from $\p{i}{j} \cap \OO{i}$ other then $i$ and $j$ lie in $\Pi^-(i,j)$, we say that candidate diagonal $(i,j)$ has \emph{negative polarity} and has $i$ as its \emph{pole}. Otherwise, if these points lie in $\Pi^+(i,j)$, we say that $(i,j)$ has \emph{positive polarity} and the pole in $j$.
	\end{definition}
	
	\begin{lemma}
		\label{lem:CandidatesDontTouchB}
		No two candidate diagonals of the same polarity can have the same point as a pole.
	\end{lemma}
	
	To prove this lemma, we use the same approach as in~\cite[Lemma 6]{savic2017faster}. The proof is deferred to Appendix.
	
	As a simple corollary of Lemma~\ref{lem:CandidatesDontTouchB}, we get that there is at most a linear number of candidate pairs.
	
	\begin{lemma}
		\label{lem:FewCandidatesB}
		There are $O(n)$ candidate pairs.
	\end{lemma}
	\begin{proof}
		Lemma~\ref{lem:CandidatesDontTouchB} ensures that there are only two candidate diagonals with poles in the same point, one having positive and one having negative polarity. Therefore, there are at most $n$ candidate diagonals of the same polarity, and, consequently, at most $2n$ candidate diagonals in total. The only other possible candidate pairs are edges, and there are exactly $n$ edges, so there can be at most $3n$ candidate pairs.	
	\end{proof}

	Finally, we combine our findings from Lemma~\ref{lem:BottleneckWithCandidateB} and Lemma~\ref{lem:FewCandidatesB}, as described in the beginning of Section~\ref{sec:MatchingsWithThreeCascadesB}, to construct Algorithm~\ref{alg:BottleneckMatchingB}.
	
	\begin{algorithm}
		\caption{Bottleneck Matching}
		\label{alg:BottleneckMatchingB}
		\begin{algorithmic}[1]
			\State Compute orbits.
			\State Compute $S^1[i,j]$ and $necessary(i,j)$, for all $i$ and $j$ such that $\pij$ is balanced, as described in Section~\ref{sec:SubproblemsB}.
			
			\State $best \leftarrow \min\{S^1[k+1,k] : k \in \p{0}{2n-1}, (k+1,k) \text{ is feasible}\}$
			
			\For {all feasible $(i,j)$}		
			\If{$necessary(i,j)$ and $\tau(i,j) \leq 2\pi/3$}
			\For {all $k \in \p{j+1}{i-1}$ such that $(j+1,k)$ is feasible}
			\State $best \leftarrow \min\{best, \max\{S^1(i,j), S^1(j+1,k), S^1(k+1,i-1)\}\}$
			\EndFor
			\EndIf
			\EndFor
		\end{algorithmic}
	\end{algorithm}
	
	\begin{theorem}
		Algorithm~\ref{alg:BottleneckMatchingB} finds the value of bottleneck matching in $O(n^2)$ time and space.
	\end{theorem}
	\begin{proof}
		The first step, computing orbits, can be done in $O(n)$ time, as described in the proof of Lemma~\ref{lem:OrbitsComplexity}.
		The second step, computing $S^1(i,j)$ and $necessary(i,j)$, for all $(i,j)$ pairs, is done in $O(n^2)$ time, as described in Section~\ref{sec:SubproblemsB}. The third step finds the minimal value of all matchings with at most one cascade in $O(n)$ time.
		
		The rest of the algorithm finds the minimal value of all $3$-cascade matchings. Lemma~\ref{lem:BottleneckWithCandidateB} tells us that there is a bottleneck matching among $3$-cascade matchings such that one inner pair of that matching is a candidate pair, so the algorithm searches through all such matchings. We first fix the candidate pair $(i,j)$ and then enter the inner for-loop, where we search for an optimal $3$-cascade matching having $(i,j)$ as an inner pair. Although the outer for-loop is executed $O(n^2)$ times, Lemma~\ref{lem:FewCandidatesB} guarantees that the if-block is entered only $O(n)$ times. The inner for-loop splits $\p{j+1}{i-1}$ in two parts, $\p{j+1}{k}$ and $\p{k+1}{i-1}$, which together with $\pij$ make three parts, each to be matched with at most one cascade. We already know the values of optimal solutions for these three subproblems, so we combine them and check if we get a better overall value. At the end, the minimum value of all examined matchings is contained in $best$, and that has to be the value of a bottleneck matching, since we surely examined at least one bottleneck matching.
		
		The algorithm uses $O(n^2)$ space for storing $n \times n$ matrices $S$ and $necessary$.
	\end{proof}
	
	Algorithm~\ref{alg:BottleneckMatchingB} gives only the value of a bottleneck matching, however, it is easy to reconstruct an actual bottleneck matching by reconstructing matchings for subproblems that led to the minimum value. This reconstruction can be done in linear time.

	\section{Points on a circle}
	\label{sec:Circle}
	
	It this section we consider the case where all points lie on a circle. Obviously, the algorithm for the convex case can be applied here, but utilizing the geometry of a circle we can do better.
	
	Employing the properties of orbits that we developed, we construct an $O(n)$ time algorithm for the problem of finding a bottleneck matching.
	
	We will make use of the following lemma.
	
	\begin{lemma}\cite{biniaz2014bottleneck}
		\label{lem:CircleEdges}
		If all the points of $P$ lie on the circle, then there is a bottleneck matching in which each point $i$ is connected either to $o(i)$ or $o^{-1}(i)$.
	\end{lemma}
	
	%
	%
	%
	%
	
	This statement implies that there is a bottleneck matching $M^E$ that can be constructed by taking alternating edges from each orbit, i.e.,~from each orbit we take either all red-blue or all blue-red edges. To find a bottleneck matching we can search only through such matchings, and to reduce the number of possibilities even more, we use properties of the orbit graph.
	
	\begin{theorem}
		A bottleneck matching for points on a circle can be found in $O(n)$ time.
	\end{theorem}
	
	\begin{proof}
		From Proposition~\ref{prp:HamiltonianPath} we know that for an arbitrary weakly connected component of the orbit graph there is a Hamiltonian path $\OL_0, \OL_1, \ldots, \OL_{m-1}$. For each $k \in \{0, \ldots, m-2\}$ there is an arc from $\OL_k$ to $\OL_{k+1}$, and those two orbits intersect each other. Since $\OL_k \preceq \OL_{k+1}$, the only edges from $\OL_k$ that intersect $\OL_{k+1}$ are blue-red edges, and only edges from $\OL_{k+1}$ that intersect $\OL_k$ are red-blue edges. Hence, $M^E$ cannot have blue-red edges from $\OL_k$ and red-blue edges from $\OL_{k+1}$. This further implies that there is $l \in \{0, 1, \ldots, m\}$ such that $\OL_0, \ldots, \OL_{l-1}$ all contribute to $M^E$ with red-blue edges and $\OL_l, \ldots, \OL_{m-1}$ all contribute to $M^E$ with blue-red edges. Let $M_l$ be the matching constructed by taking red-blue edges from $\OL_0, \ldots, \OL_{l-1}$, and blue-red edges from $\OL_l, \ldots, \OL_{m-1}$.
		
		For each $l$, the value of $M_l$ can be obtained as $\max \{RB_l, BR_l\}$, where $RB_l$ is the length of the longest red-blue edge in $\OL_0, \ldots, \OL_{l-1}$, and $BR_l$ is the length of the longest blue-red edge in $\OL_l, \ldots, \OL_{m-1}$. The computation of sequences $RB$ and $BR$ can be done in $O(n)$ total time, since $RB_l$ is maximum of $RB_{l-1}$ and the longest red-blue edge in $\OL_{l-1}$, and $BR_l$ is maximum of $BR_{l+1}$ and the longest blue-red edge in $\OL_l$. After we compute these sequences, we compute the value of $M_l$ for each $l$, and take the one with the minimum value, which must correspond to a bottleneck matching.
		
		We first compute orbits and Hamiltonian paths in $O(n)$ time (Lemma~\ref{lem:OrbitsComplexity} and \ref{lem:OrbitGraphComplexity}). Next, we compute the longest red-blue and blue-red edge in each orbit, which we then use to compute $RB_l$, $BR_l$, and $M_l$, for each weakly connected component of the orbit graph, and finally $M^E$, as we just described. Each step in this process takes at most $O(n)$ time, so the total running time for this algorithm is $O(n)$ as well.
	\end{proof}

\section*{Acknowledgments}
We are grateful to the anonymous referees, whose useful and detailed comments improved our paper.

	\bibliographystyle{unsrt}
	\bibliography{references}
	
	\appendix
	
	\section{Appendix}
	
	\begin{proof}(of Lemma~\ref{lem:PiHalfB})
		\figX{fig:PiHalfB}{(a) Matching before the transformation. (b) Matching after the transformation.}{%
			\parbox{4cm}{\centering\includegraphics{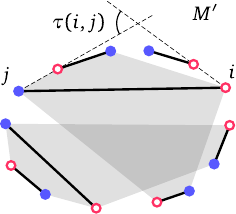}\\(a)}
			\hspace{1cm}%
			\parbox{4cm}{\centering\includegraphics{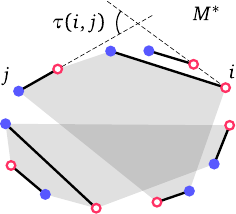}\\(b)}
		}%
		
		Let us suppose that there is no such matching. Let $M'$ be a bottleneck matching with the least number of diagonals. By the assumption, there is a diagonal $(i,j) \in M'$ such that $\tau(i,j) \leq \pi/2$, see Figure~\ref{fig:PiHalfB}(a). By Proposition~\ref{prp:MatchingWithEdgesBalanced} we can replace all pairs from $M'$ lying in $\pij$, including the diagonal $(i,j)$, with the matching containing only edges, and by doing so we obtain a new matching $M^*$, see Figure~\ref{fig:PiHalfB}(b).
		
		The longest distance between any pair of points from $\pij$ is achieved by the pair $(i,j)$, so $bm(M^*) \leq bm(M')$. Since $M'$ is a bottleneck matching, $M^*$ is a bottleneck matching as well, and $M^*$ has at least one fewer diagonal than $M'$, a contradiction.
	\end{proof}
	
	\begin{proof}(of Lemma~\ref{lem:ThreeCascadesB})
		Let $M$ be a matching provided by Lemma~\ref{lem:PiHalfB}, with turning angles of all diagonals greater than $\pi/2$. There cannot be a region bounded by four or more diagonals of $M$, since if it existed, the total turning angle would be greater than $2\pi$. Hence, $M$ only has regions with at most three bounding diagonals. Suppose there are two or more $3$-bounded regions. We look at two arbitrary such regions. There are two diagonals bounding the first region and two diagonals bounding the second region such that these four diagonals are in cyclical formation, meaning that each diagonal among them has other three on the same side. Applying the same argument once again we see that this situation is impossible because it yields turning angle greater than $2\pi$. We conclude that there can be at most one $3$-bounded region.
	\end{proof}

	\begin{proof}(of Lemma~\ref{lem:BottleneckWithAllNecessaryB})
		Take any $3$-cascade bottleneck matching $M$. If it has an inner pair $(i,j)$ that is not necessary, then (by definition) there is a solution to $\Matching^1(i,j)$ that does not contain the pair $(i,j)$ and has at most one cascade. We use that solution to replace all pairs from $M$ that are inside $\pij$, and thus obtain a new $3$-cascade matching that does not contain the pair $(i,j)$. Since $M$ was optimal and there was at most one cascade inside $\pij$, pairs that were replaced are also a solution to $\Matching^1(i,j)$, so the new matching must have the same value as the original matching. And since there is no bottleneck matching with at most one cascade, the new matching must be a bottleneck $3$-cascade matching as well. We repeat this process until all inner pairs are necessary. The process has to terminate because the inner region is getting larger with each replacement.
	\end{proof}
	
	\begin{proof}(of Lemma~\ref{lem:BottleneckWithCandidateB})
		Lemma~\ref{lem:BottleneckWithAllNecessaryB} provides us with a $3$-cascade matching $M$ whose every inner pair is necessary. There are at least three inner pairs of $M$, so at least one of them has turning angle at most $2\pi/3$. Otherwise, the total turning angle would be greater than $2\pi$, which is not possible. Such an inner pair is a candidate pair.
	\end{proof}

	\begin{proof}(of Lemma~\ref{lem:CandidatesDontTouchB})
		\fig{fig:CandidatesDontTouchB}{Two candidate diagonals of equal polarity cannot have the same pole.}{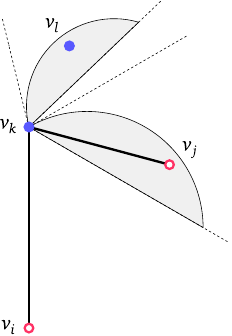}%
		Let us suppose the contrary, that is, that there are two candidate diagonals of the same polarity with the same point as the pole. Assume, w.l.o.g., that $(i,k)$ and $(j,k)$ are two such candidate diagonals, $i \neq j$, both with positive polarity, each having its pole in $k$. Since both $(i,k)$ and $(j,k)$ are feasible pairs, $i$, $j$ and $k$ belong to the same orbit. W.l.o.g., we assume that the order of points in the positive direction is $i$ -- $j$ -- $k$, that is, $j \in (\p{i}{k} \cap \OO{k}) \setminus \{i,k\}$, see Figure~\ref{fig:CandidatesDontTouchB}.	
		
		The region $\Pi^+(i,k)$ lies inside the wedge with the apex $v_k$ and the sides at the angles of $\pi/3$ and $2\pi/3$ with the line $v_kv_i$. Similarly, $\Pi^+(j,k)$ lies inside the wedge with the apex $v_k$ and the sides at the angles of $\pi/3$ and $2\pi/3$ with the line $v_kv_j$. This means that $\Pi^+(i,k)$ and $\Pi^+(j,k)$ have no points in common other than $v_k$.
		
		Since $(j,k)$ is a diagonal, there is $l \in (\p{j}{k} \cap \OO{k}) \setminus \{j, k\}$. But $l \in (\p{i}{k} \cap \OO{k}) \setminus \{i,k\}$ as well, meaning that $l \in \Pi^+(i,k) \cap \Pi^+(j,k)$, which contradicts the conclusion of the previous paragraph.
	\end{proof}

\end{document}